\documentclass[12pt,authoryear]{elsarticle}
\makeatletter
\def\ps@pprintTitle{%
 \let\@oddhead\@empty
 \let\@evenhead\@empty
 \def\@oddfoot{\centerline{\thepage}}%
 \let\@evenfoot\@oddfoot}

\usepackage[page]{appendix} 

\usepackage{amssymb}
\usepackage{amsmath}
\usepackage{amsthm}
\newtheorem{theorem}{Theorem}

\newenvironment{proofsketch}{%
  \proof}{\endproof}

\usepackage{hyperref}
\hypersetup{
    colorlinks=true,
    linkcolor=blue,
    filecolor=magenta,      
    urlcolor=cyan,
}
\usepackage{graphicx}
\graphicspath{ {./images/} }
\usepackage{wrapfig}
\usepackage{svg}

\usepackage{tikz}
\newcommand*\circled[1]{\tikz[baseline=(char.base)]{
            \node[shape=circle,draw,inner sep=2pt] (char) {#1};}}
\usepackage{float}

\begin{document}

\begin{frontmatter}



\title{Alignment Problems With Current Forecasting Platforms}


\author[1]{Nu\~{n}o Sempere\corref{cor1}}
\author[2]{Alex Lawsen}

\address[1]{Quantified Uncertainty Research Institute, Vienna}
\address[2]{Kings Maths School, United Kingdom}

\cortext[cor1]{Corresponding author. E-mail address: \url{nuno@quantifieduncertainty.org}}

\begin{abstract}
    We present alignment problems in current forecasting platforms, such as Good Judgment Open, CSET-Foretell or Metaculus. We classify those problems as either reward specification problems or principal-agent problems, and we propose solutions. For instance, the scoring rule used by Good Judgment Open is not proper, and Metaculus tournaments disincentivize sharing information and incentivize distorting one's true probabilities to maximize the chances of placing in the top few positions which earn a monetary reward. We also point out some partial similarities between the problem of aligning forecasters and the problem of aligning artificial intelligence systems.
\end{abstract}





\end{frontmatter}



\section{Introduction}
\label{Introduction}

\subsection{Motivation: The importance of alignment problems for the forecasting ecosystem}\label{Motivation}

Forecasting systems and competitions such as those organized by Good Judgment, Metaculus or CSET-foretell have been used to inform probabilities of nuclear war (\cite{Rodriguez}), the probability of different coronavirus scenarios (\cite{covidrecovery}), the chances of each presidential candidate winning the US election (\cite{USelections}), the probability of heightened geopolitical tensions or conflict with China (\cite{SouthChina}), the likelihood of global catastrophic risk (\cite{ragnarok}), or the rates of AI progress (\cite{aiprogress}). Further, these probabilities aim to be action guiding, that is, to be accurate and respected enough to influence real world decisions:

\begin{quote}
    ``...the tools are good enough that the remarkably inexpensive forecasts they generate should be on the desks of decision makers, including the president of the United States" (\cite{Tetlock}, Chapter 9)
\end{quote}

However, these forecasting competitions sometimes inadvertently provide forecasters with incentives not to reveal their best forecasts. As as a highlight of the paper, in \S\ref{GJ not proper}, we prove that the scoring rule used in Good Judgement Open, CSET-foretell, and other Cultivate Labs platforms is not proper. That is, in some scenarios, a forecaster with correct probabilistic beliefs can input a much higher probability than would reflect their true beliefs into the platform, and still obtain a better score in expectation, and this effect is quantitatively large. 

Notably, Good Judgement draws from the top 2\% of forecasters from Good Judgement Open, who are then dubbed Superforecasters™, which introduces further incentive distortions.

The incentive problems we identify are not only problematic because score-focused forecasters might exploit them, but also because platform users might interpret imperfect reward schemes as feedback, and because the flawed incentive schemes will fail to incentivize, identify and reward the best forecasters.

For an overview of the broader literature around incentives for forecasters, see the recent literature review in (\cite{Witkowski}). 


\subsection{Overview of the paper}

Section \S\ref{Motivation} provides the motivation for this paper, and section \S\ref{Alignment terminology} introduces the alignment terminology we use as an organizing principle for the rest of the paper. Sections \S\ref{Outer alignment incentive Problems} and \S\ref{Inner alignment incentive Problems} outline the alignment problems which we identify, and Appendix \S\ref{Simulations} provides some numerical simulations to support our points. In section \S\ref{Solutions} we propose possible solutions. Section \S\ref{Conclusion} concludes, and draws parallels to the broader artificial intelligence alignment problem. 
\newpage

\newpage
\subsection{Alignment terminology}\label{Alignment terminology}

The process of creating a forecasting system comprises several distinct steps: 

\begin{enumerate}
    \item Broader society, with its own goals, spawns a forecasting system, whose goals are to obtain accurate probabilities.
    \item The forecasting system chooses a formal objective function which operationalizes ``obtain accurate probabilities".
    \item Forecasters then maximize their actual reward, which might depend on their own preferences and inclinations.
\end{enumerate}

\begin{wrapfigure}{l}{0.4\textwidth}
    \centering
    \includegraphics[width=0.35\textwidth]{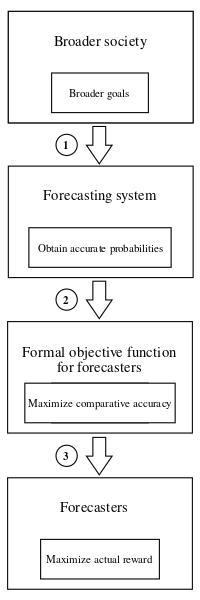}
\end{wrapfigure}

An alignment failure is possible at each of those steps. For instance, a failure in the first step, where broader society's goals are not aligned with a forecasting system's goals, might correspond to obtaining forecasts about Ebola without being cognizant that different forecasts may provoke different responses, and hence lead to a different number of deaths, and that the prediction which minimizes deaths might not be the one which is most accurate. A failure in the second step, where the goals of the forecasting system are not aligned with its scoring rule, might correspond to choosing a reward function which, when maximized, doesn't result in optimal forecasts, for example, a function which rewards successfully convincing other forecasters of false information: this maximizes the liar's comparative accuracy, but might decrease the accuracy of the broader scoring system. A failure in the third step, where the goals of the forecasters are not aligned with the goals of the formal scoring rule might correspond to a bug in the implementation of the scoring rule, so that forecasters optimize their score of the unaligned buggy implementation.

\cite{Hubinger} distinguish between inner and outer alignment. Outer alignment refers to choosing an objective function which to optimize, and making sure that the gap between this objective function and one's own true goals is as small as possible. This has some similarities to our first two steps above: making sure that we want to optimize for obtaining accurate probabilities, and making sure that our scoring rule leads to accurate probabilities. For instance, if we care about saving lives, we would have to check that having better forecasts leads to more lives saved, and that the scoring rule, if optimized, leads to more accurate forecasts.

Inner alignment then refers to making sure that an optimizer is trying to optimize a previously chosen base objective function, when this optimizer is also being optimized for. Normally this refers to making sure that the reward function learnt by a reinforcement learner during training closely approximates a base objective function, even in a different environment outside of training. In the case of a forecasting system, this bears some resemblance to making sure that forecasters maximize their score according to a previously chosen scoring rule, for instance by paying forecasters more money the lower their Brier score is.

Nonetheless, human forecasters have some differences with reinforcement learning systems, in that they are not quite trained by forecasting systems from scratch, but merely repurposed or incentivized. For this reason, it is perhaps more accurate to say that the first two steps are a akin to a reward specification problem, and the last step is akin to a principal-agent problem. 

The focus of this paper is on failures for forecasting platforms related to steps \circled{2} and \circled{3}, respectively covered in sections \S\ref{Outer alignment incentive Problems} and \S\ref{Inner alignment incentive Problems}. 


\section{Reward specification problems $\approx$ Outer alignment problems} 
\label{Outer alignment incentive Problems}

This question discusses cases where, as individual forecasters maximize their comparative accuracy score\footnote{e.g., their relative Brier score}, or some other score determined by a forecasting platform, problems arise where this does not maximize the accuracy of the whole forecasting system. That is, the forecasting platform chose some reward function, forecasters are optimizing that reward function, but it turns out that the reward function fails to capture some aspect which the creators or clients of the forecasting platform also care about.

An classical example of this would be a Sybil attack: if a forecaster created many puppet accounts on Good Judgment Open, CSET or Metaculus, and used them to make many bad forecasts on the questions she had made predictions on, her comparative accuracy score (and the comparative accuracy score of those who predicted the same questions as her) would increase. 

As long as set of questions she choose to forecast on were relatively unique, she would benefit. She would still be maximizing her comparative accuracy score, but that score would have ceased to be related to the broader objectives of the forecasting system.

\subsection{The scoring rule incentivizes forecasters to selectively pick questions}
\label{Selectively pick questions}

Forecasters seeking to obtain a good score are incentivized to selectively pick questions, and in some cases, are better off not making predictions \textit{even if they know the true probability exactly.} Tetlock mentions a related point in one of his ``Commandments for superforecasters'' (\cite{Tetlock}, pp. 277-278): ``Focus on questions where your hard work is likely to pay off''. Yet if we care about making better predictions of things we need to know the answer to, the skill of ``trying to answer easier questions so one scores better" is not a skill we should reward, let alone encourage the development of.

For the case where the forecasting system rewards the Brier score\footnote{In Brier Scoring, lower scores are better}, if a forecaster has a brier score $p \cdot (1-p)$, then they should not make a prediction on any question where the probability is between $p$ and $(1-p)$, \textit{even if they know the true probability exactly.} 

\begin{theorem}
    A forecaster wishing to obtain a low average Brier score, and who has\footnote{Either currently, or in expectation} a Brier score of $p \cdot (1-p)$ (with $p\le(1-p)$ without loss of generality) should only make predictions in questions where the probability is lower than $p$ or higher than $(1-p)$.
\end{theorem}
\begin{proof}
    Suppose that the forecaster makes a prediction on a binary question which has a probability $q\le 0.5$ of resolving positively, with $p < q < (1-p)$. Then the expected Brier score is:

    \begin{equation}
        E[Score] = q \cdot (1-q)^2 + (1-q) \cdot (0-q)^2 = q \cdot (1-q) \cdot (1-q + q) = q \cdot (1-q)
    \end{equation}

    As $f(x) = x \cdot (1-x)$ is strictly increasing from $0$ to $0.5$, if $p<q$, then $p \cdot (1-p )< q \cdot (1-q)$. Hence, their expected score on the question is greater than their current score, and they should not predict. For the case where $q>0.5$, consider that the question has a probability $q'$ of resolving negatively and swap $p$ and $(1-p)$.

\end{proof}

In fact, if a forecaster has so far performed better than random guessing, there exists a range of probabilities which, if used, are guaranteed to hurt their score not only in expectation, but regardless of the outcome of the event. This is stated formally as follows.

\begin{theorem}
    A forecaster with brier score $b^2$ who forecasts that an event has probability $p$ is guaranteed to end up with a worse brier score, regardless of the outcome of the event, if $b<p<(1-b)$.
\end{theorem}
\begin{proof}
      If $p<0.5$, the best score the player can achieve is $p^2$, when the event does not occur. But $p^2>b^2$ as $p>b$, so the player's score will be worse.

   If $p>0.5$, the best score the player can achieve is $(1-p)^2$, when the event occurs. 
\begin{equation}
\begin{split}   
(1-p)>1-(1-b)\\
\implies(1-p)>b\\
\implies(1-p)^2>b^2\\
\end{split}
\end{equation}  
   Hence, in this case the player's score will also be worse.
    
\end{proof}

Some competitions don't reward the Brier score, but the relative Brier score (or some combination of both). The relative Brier score is defined as the difference between the forecaster's Brier scores, and the aggregate's Brier scores. As before, a lower score is better. 

As in the previous case, forecasters should sometimes not predict on some questions, even if they know the probability exactly. This is not necessarily a problem, as it might lead to better allocation of the forecaster's attention, but can be. 

\begin{theorem}
    A forecaster seeking to obtain a low average relative Brier score score, and who has\footnote{Again, either currently or in expectation} a relative Brier score of $r$, should only make predictions in questions where:
\begin{equation}
    E[\textnormal{the forecaster's Brier score}] - E[\textnormal{Brier score of the aggregate}] < r
\end{equation}
\end{theorem}
\begin{proof}
\begin{equation}
\begin{split}
&E[\textnormal{the forecaster's Brier score}] - E[\textnormal{Brier score of the aggregate}] \\
&= E[\textnormal{the forecaster's Brier score} - \textnormal{Brier score of the aggregate}] \\
&= E[\textnormal{relative Brier score}]
\end{split}
\end{equation}
and the forecaster should only predict if $E[\textnormal{relative Brier score}]<r$. Otherwise, predicting degrades their relative Brier score, in expectation. 
\end{proof}

This can be particularly problematic for questions where the community is confident, e.g., a question for which the aggregate forecast is $5\%$. Suppose that a forecaster knew with certainty that the event would not happen, and forecasted $0\%$. Then the resulting relative Brier Score would be $0 - 0.05^2/2 = - 0.00125$, which by inspection is poor; the authors observe that good forecasters tend to have relative Brier scores downwards of $-0.05$.

Note that this is only a problem if one is trying to maximize the mean relative Brier, that is, the average difference between her forecasts and the aggregate's. The problem disappears if one is trying to maximise the sum of relative Brier scores (as in the case of CSET-foretell), although the discussion below should be noted.

\subsection{The scoring rule incentivizes forecasters to just copy the community on every question.}
\label{Copy the opinion}

In scoring systems which more directly reward making many predictions, such as the Metaculus scoring system where in general one has to be both wrong and confident to lose points, predictors are heavily incentivised to make predictions on as many questions as possible in order to move up the leaderboard. In particular, a strategy of ``predict the community median with no thought'' could---in the expert opinion of one of the authors, a top 50 Metaculus forecaster himself---see someone rise to the top 100 within a few months of signing up. Metaculus allows advanced users to see other forecaster's track records by spending ``tachyons'', from which the authors found that at least one of the top 10 predictors on the leaderboard has performed worse than the community prediction\footnote{This is different from the ``Metaculus prediction'', which uses a proprietary algorithm to aggregate predictions. The community prediction merely averages them (weighted by recency).} on a per-question basis for both discrete and continuous questions.

However, if the main value of winning ``Metaculus points'' is personal satisfaction, then predicting exactly the community median is unlikely to keep participants entertained for long. New users predicting something fairly close to the community median on lots of questions, but updating a little bit based on their own thinking, is arguably not a problem at all, as the small adjustments may be enough to improve the crowd forecast, and the high volume of practice that users with this strategy experience might lead to rapid improvement.

This incentive to predict on all questions can also arise when scoring is based on the total (cumulative) relative Brier score. 

\begin{theorem}
    A forecaster seeking to optimize her relative Brier score should predict the average community forecast instead of not predicting on a question, even if she doesn't know anything about the question.
\end{theorem}
\begin{proof}
The relative Brier score is calculated as an average relative score over all other predictors, i.e.

\begin{equation}
RBS(p_0) =  Brier(p_0) - \frac{1}{n}  \sum_{i=1}^{n} Brier(p_i)
\end{equation}

The community prediction is given by $\displaystyle \frac{1}{n}  \sum_{i=1}^{n} p_i$, so the forecaster's Brier score if she predicts the community prediction is

\begin{equation}
RBS(p_0) =  \displaystyle Brier\left(\frac{1}{n}  \sum_{i=1}^{n} p_i\right) - Brier\left( \frac{1}{n}  \sum_{i=1}^{n} p_i\right)
\end{equation}

However, the Brier score is strictly increasing and convex, so 
\begin{equation}
    Brier\left( \frac{1}{n}  \sum_{i=1}^{n} p_i\right) \le  \frac{1}{n} \sum_{i=1}^{n} Brier(p_i)
\end{equation}

Thus the forecaster's relative Brier score  will always be negative\footnote{i.e., better} (or zero if all predictions are the same.) And thus predicting the community average is more advantageous than leaving the question blank.
\end{proof}

In practice, this makes community predictions more ``sticky'', that is, slower to adapt to new information. This is because new forecasts with new information are averaged with the old forecasts, of which there is an incentive to be many. This would not be a problem if forecasters updated often, but in practice forecasters who are lazy enough to forecast the community average on many questions in order to gain some marginal relative Brier or Metaculus points also don't update their forecasts that often.



\subsection{The scoring rule incentivizes forecasters not to share information and to produce corrupt information}
\label{Corrupt info}

In a forecasting tournament which rewards comparative accuracy, there is a disincentive to sharing information, because other forecasters can use it to improve relative standing. This disincentive includes both not sharing information and providing misleading information. Perhaps for this reason, most forecasts on Good Judgement Open are accompanied by blank comments. 

As a counterpoint, other forecasters can and will often point out flaws in one's reasoning if one gives an unconvincing rationale. Further, regardless of the rationale one writes, in Good Judgment Open and CSET, other forecasters can see one's probabilities. Thus, seeing a blank forecast rationale accompanied by a forecast of 100\% by a forecaster with a good track record can be enough to infer that this forecaster has private information. 

Publishing misleading information does not seem to be an urgent problem today. However, this might only be the case because forecasting competitions are currently relatively niche and small. If they grew further, they might encounter similar problems to PredictIt, where market participants have on occasion created fake polls which confused election bettors.

Two notable cases from the PredictIt community are the cases of Delphi Analytica, and  CSP polling. With regards to Delphi Analytica, \cite{Fake polls real problem} explains in a FiveThirtyEight article:

\begin{quote}
    ``Delphi Analytica released a poll fielded from July 14 to July 18. Republican Kid Rock earned 30 percent to Sen. Debbie Stabenow’s 26 percent. A sitting U.S. senator was losing to a man who sang the lyric, “If I was president of the good ol’ USA, you know I’d turn our churches into strip clubs and watch the whole world pray.”
    
    McDonald believes that `Jones' and whoever may have helped him or her did so for two reasons. The first: to gain notoriety and troll the press and political observers. (The message above seems to support that theory.) The second: to move the betting markets. That is, a person can put out a poll and get people to place bets in response to it — in this case, some people may have bet on a Kid Rock win — and the poll’s creators can short that position (bet that the value of the position will go down). In a statement, Lee said Delphi Analytica was not created to move the markets. Still, shares of the stock for Michigan’s 2018 Senate race saw their biggest action of the year by far the day after Delphi Analytica published its survey.

    The price for one share — which is equivalent to a bet that Stabenow will be re-elected — fell from 78 cents to as low as 63 cents before finishing the day at 70 cents. (The value of a share on PredictIt is capped at \$1.) McDonald argued that the market motivations were likely secondary to the trolling factor, but the mere fact that the markets can be so easily manipulated is worrisome.''
\end{quote}

With regards to CSP polling, \cite{Yeargain} explains in the Michigan Law Review:
\begin{quote}
    ``a PredictIt user seeking to purchase a futures contract on the outcome of the Republican primary in Alabama’s 2017 special U.S. Senate election who comes across a poll predicting a result of that exact election, allegedly conducted by CSP Polling, might reasonably consider that poll in their purchasing decision – even if they do not know that CSP lacks a track record or any indicia of reliability.  And given the speed with which PredictIt users buy and sell contracts, a user seeing this information might reasonably conclude that if she is to use this information to her benefit, she needs to act quickly.

    CSP Polling – which, according to University of Florida political science professor Michael McDonald and Jeff Blehar of the National Review, stands for `Cuck Shed Polling' – alleged that it conducted polls in the 2017 special congressional election in Montana, the special congressional election in Georgia, and the Virginia Democratic primary for Governor. Even after being identified in FiveThirtyEight as a fake pollster, CSP Polling continued to release polls, though the seriousness of the poll “releases” noticeably deteriorated in the year that followed.''
\end{quote}

For an example within the forecasting space, the Forecasting AI Progress Tournament (\cite{aiprogress}) saw very few forecasts in its first round, so organizers decided to require the publication of at least three comments per participant in subsequent rounds. However, participants are still cognizant that sharing their insights reduces the advantage they have over competitors, and thus their expected share of the \$50,000 prize. A top 10 forecaster shares ``I personally have made less comments on Metaculus over time (especially on tournaments with monetary rewards) due to realizing that it's better for me not to share.''

A different way to look at the perverse incentives of non-cooperative scoring rules is to imagine a population which is differential altruistic, as is the case in Metaculus or in CSET-Foretell. More altruistic participants who share information benefit the less altruistic participants who do not. With the passage of time, altruistic forecasters would be systematically disadvantaged, and less altruistic forecasters would rise to the top. Of course, altruistic forecasters are not necessarily naïve, and foreseeing this dynamic, may choose to not share information either, or to only share information among themselves.

If we look at Metaculus' leaderboard, top 10 forecasters leave comments in respectively 5.5\%, 16\%, 21\%, 6\%, 12\%, 60\%\footnote{Not a typo}, 1.5\%, 0.7\% 6\% and 15\% of the questions they forecast on. This percentage seems low, and it seems natural to hypothesize that it would be higher if forecasters were not incentivized to make it lower.

If forecasting competitions were to be expanded to take a broader role in society, if monetary prizes were to significantly increase, or if the forecasting community were to grow in size and become less tightly knit, the above problems might be exacerbated.1

\section{Principal-agent problems $\approx$ Inner alignment incentive problems} 
\label{Inner alignment incentive Problems}

Forecasters sometimes deeply care about their score on forecasting platforms, spending long and mostly unpaid hours forecasting on them. Sometimes, they care about reward maximization explicitly; for example, forecaster @yagudin\footnote{One of the authors is in fact in a forecasting team with @yagudin, but doesn't share that amoral  maximizing perspective to such an extent.} (who is in the top 1\% of CSET-foretell forecasters) writes \textit{``If this question resolves positively I will get a lot of points. This question could only resolve negatively in 5 years, likely, I wouldn't care about my Foretell predictions then"} as a comment under ``Will the U.S. government file an antitrust lawsuit against Apple, Amazon, or Facebook between January 20, 2021 and January 19, 2025?". In this case, the forecaster has ceased to maximize expected comparative predictive accuracy and ceased to report his true probabilities, and has started maximizing a different reward. 

In the previous section, we considered forecasters who were dutifully optimizing their expected accuracy (e.g., their Brier score, or the difference between their Brier score and that of their competitors). In particular, forecasters were still reporting their true probabilities throughout. This still led to problems because there is some difference between ``maximize your comparative accuracy'' and ``produce useful information''. In this section, we consider problems where forecasters cease to report their true probabilities, even when forecasting competitions organizers wish they hadn't, or think they haven't.


\subsection{If scoring is translated to discrete prizes, this creates an incentive to distort forecasts} 
\label{Discrete prizes}

If a forecasting tournament offers a prize to the top X forecasters, the objective ``be in the top X forecasters" differs from ``maximize predictive accuracy". The effects of this are greater the smaller the number of questions.

For example if only the top forecaster wins a prize, forecasters might want to predict a surprising scenario, because if it happens they will reap the reward, whereas if the most likely scenario happens, everyone else will have predicted it too.



This effect is also present if one considers acquiring the title of "Superforecaster" as a reward or as an objective in itself. One might consider it so because it comes with a degree of prestige and career capital. This objective can be achieved by by reaching the top 2\% of forecasters and having over 100 predictions in Good  Judgement Open. But note that the objective of ``place in the top 2\%" is different from ``be a s accurate as possible."


\subsubsection{Optimal distortion in a tournament with continuous questions}
To attempt to quantify these effects, we ran some numerical simulations for continuous questions. In this section, we consider a perfect predictor in a tournament with a varied pool of 30 players. This perfect predictor can either predict the true continuous distribution from which the question resolution will be drawn, or manipulate its mean and standard deviation. 

More detailed specification details, as well as more complex simulations and simulations for binary rather than continuous questions can be found in Appendix \S\ref{Simulations}.  

In Figures \ref{MeanDistortion} and \ref{SdDistortion}, we find that, given Metaculus's scoring rule,  after roughly 10 questions, the incentive to distort either the mean or  the standard deviation for one question is much reduced. In Figure \ref{MeanAndSdDistortion} we then consider a small tournament of 5 questions, and observe that the optimal distortion is attained when distorting the mean by $50$  units and the standard deviation by $12$ units. This results in a probability of placing in the top 3 of $16.2\%$, as opposed to $13.46\%$ in the case where there is no distortion. 

\begin{figure}[H] 
    \includegraphics[width=12cm]{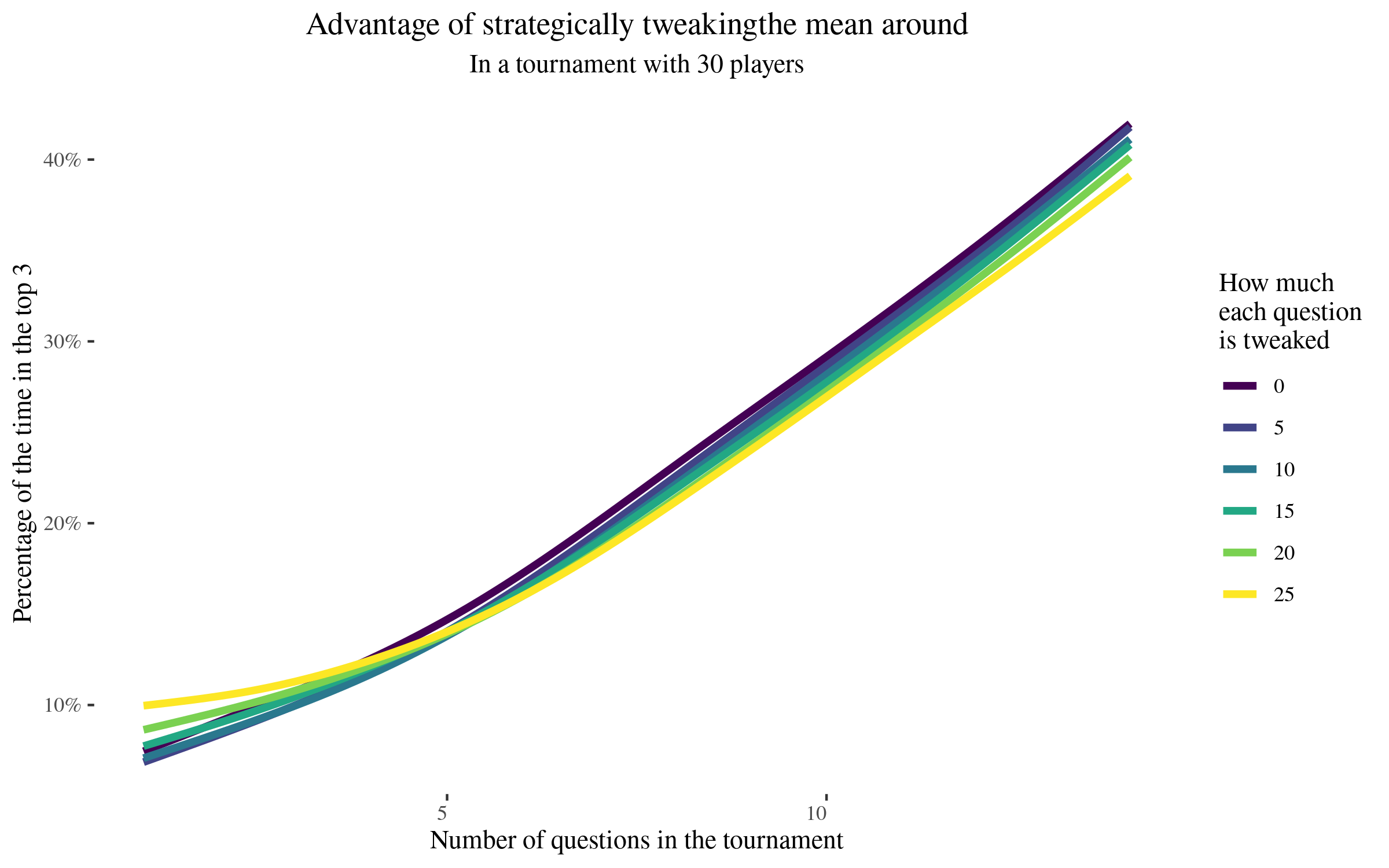}
    \centering
    \caption{Distorting the mean for one question in a five question tournament with 30 other participants. 2000 simulations for each point.} \label{MeanDistortion}
\end{figure}

\begin{figure}[H] 
    \includegraphics[width=12cm]{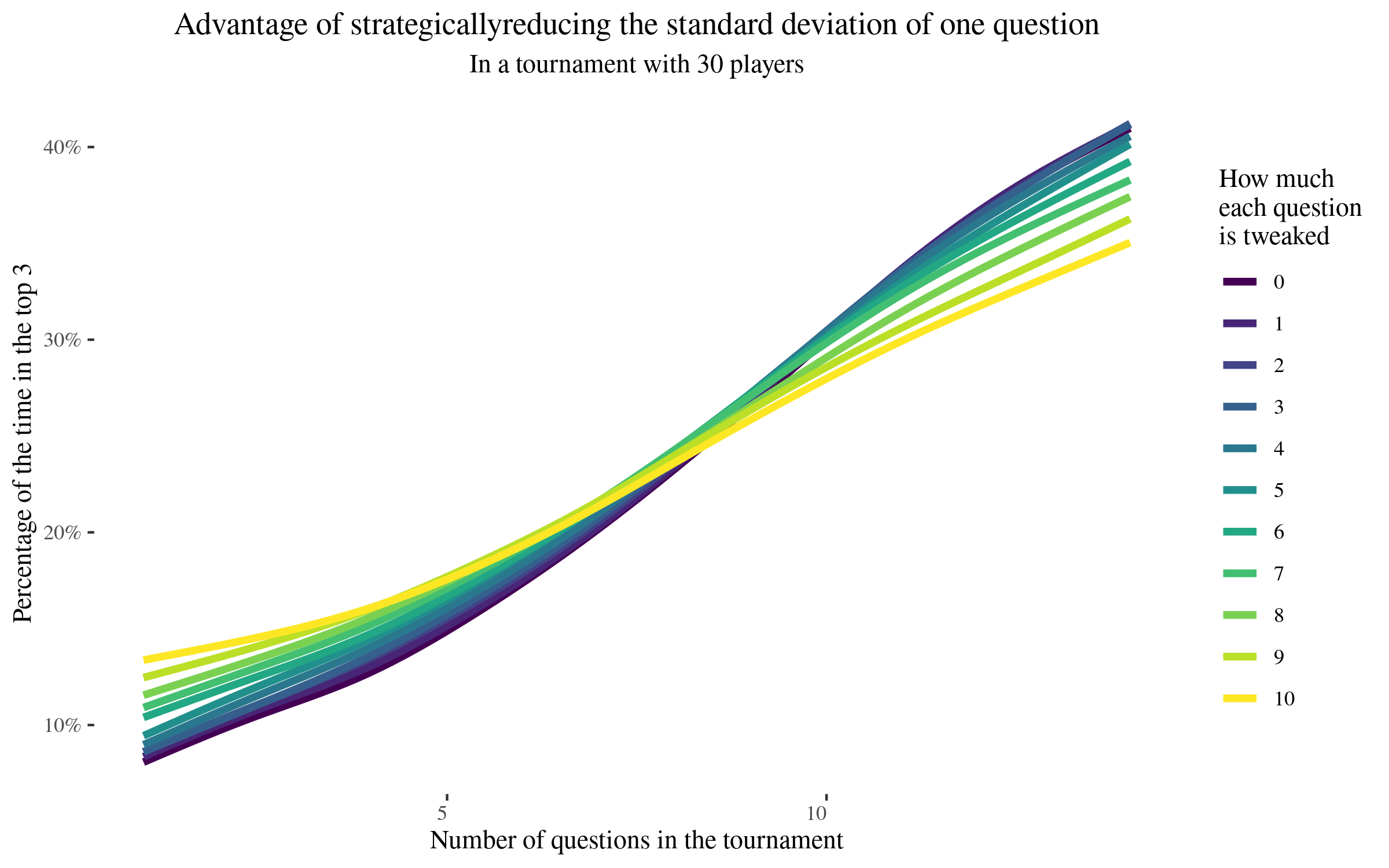}
    \centering
    \caption{Distorting the standard deviation for one question in a five question tournament with 30 other participants. 2000 simulations for each point.}\label{SdDistortion}
\end{figure}

\begin{figure}[H] 
    \includegraphics[width=12cm]{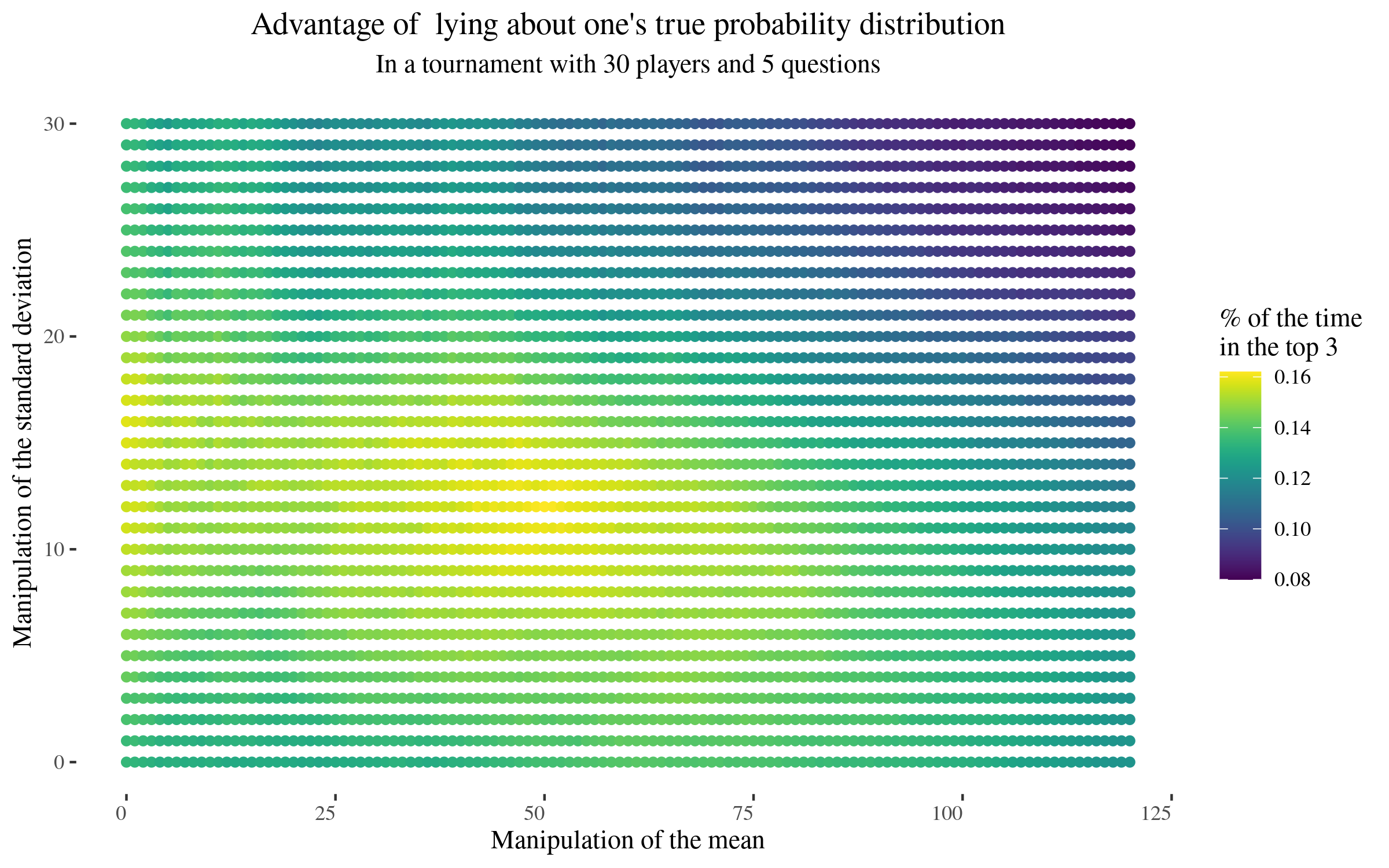}
    \centering
    \caption{Distorting both the mean and the standard deviation for one question in a five question tournament with 30 other participants. 5000 simulations for each point.} \label{MeanAndSdDistortion}
\end{figure}

Note that there are decreasing marginal returns to distorting any one question, so in practice, better results might be achieved by distorting more than one question by different amounts. Thus, the above can be thought of as a lower bound for the amount of distortion, not an upper bound.


For a result closer to an upper bound, we run the same simulations as above, but we also allow for variation in the number of questions the predictor is allowed to distort beyond just one. We find that that the maximum is attained when two questions are distorted, with a distortion of the mean of 0 units and a distortion of the standard deviation of 10 units. In that case, the predictor places in the top 3 $17.3\%$ of the time, in contrast with the previous $13.46\%$ in the case where there is no distortion and $16.2\%$ when only distorting one question.

\begin{figure}[H] 
    \includegraphics[width=12cm]{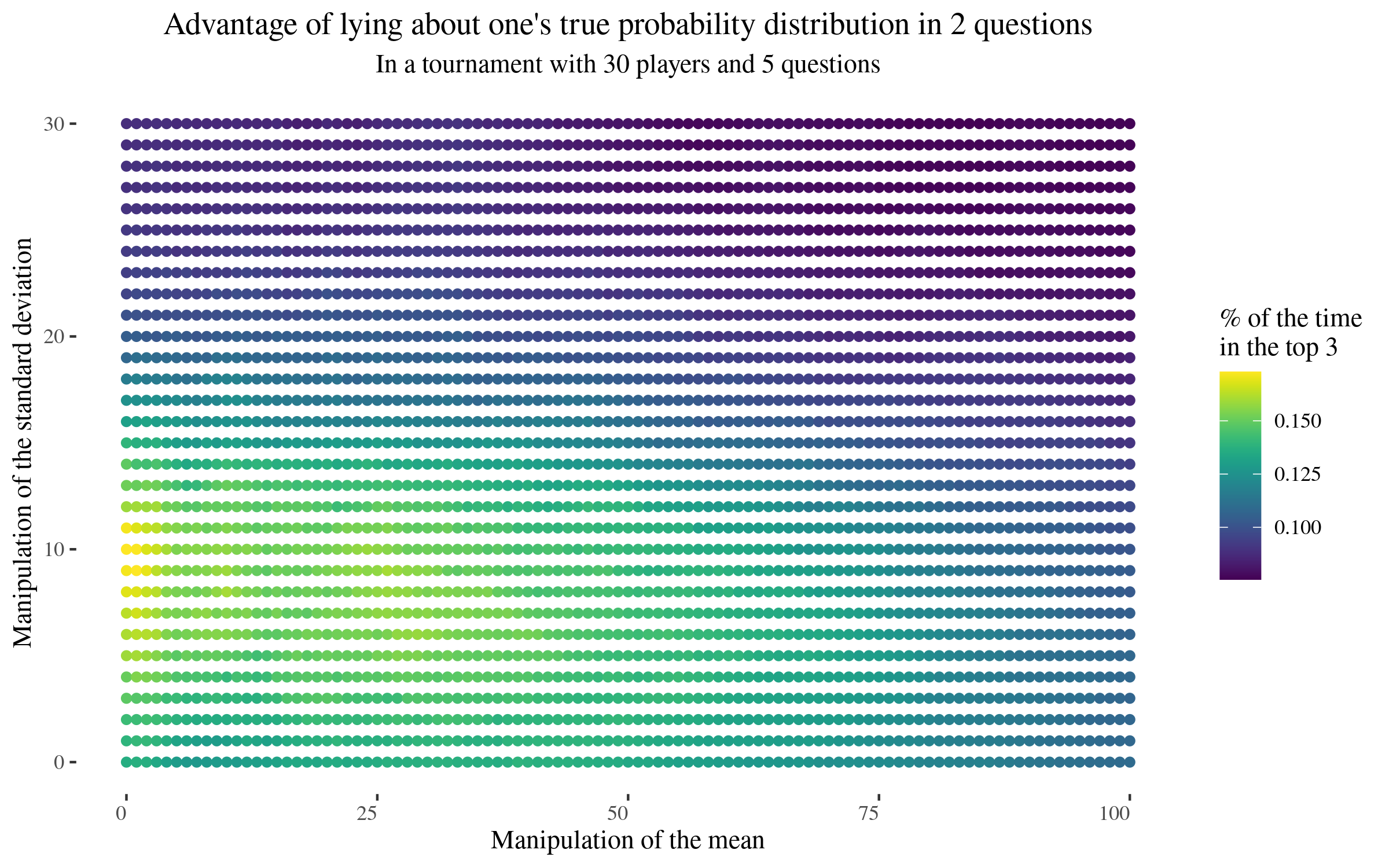}
    \centering
    \caption{Distorting both the mean and the standard deviation for two questions in a five question tournament with 30 other participants. 5000 simulations for each point.} \label{MeanAndSdDistortion2questions}
\end{figure}

The point that forecasters seeking to place at the top of a leaderboard will have different incentives to purely maximising their score is known among forecasting researchers and practitioners; see e.g., (\cite{Lichtendahl}) for a proof in the two forecaster case. However, we are not aware of previous literature which considers numeric simulations of continuous questions to quantify the incentive to distort one's own forecasts.

\subsection{The participation-rate weighted Brier score is not a proper scoring rule}
\label{GJ not proper}

We define ``participation-rate weighted Brier score'' as follows: the Brier score integrated in time multiplied by the participation rate. The participation rate is the percentage of time that the forecaster had an active prediction, from the beginning of the question until question resolution. 

More formally, if a question opens at time $t_0$, a forecaster starts forecasting at time $t_1$, the question closes\footnote{In this setup, forecasters can't withdraw from a question, which means that if they stop updating their forecasts, their relative score suffers as a result} at time $t_2$, and the forecaster's prediction at time $t$ is $P(t)$, then ``participation-rate weighted Brier score'' is:
\begin{equation}\label{PWBS definition}
    PWBS(p) = \frac{1}{t_2-t_0} \cdot \int_{t_1}^{t_2} Brier(p(t)) dt
\end{equation}

In this section, we prove that the ``participation-rate weighted Brier score'' is not a proper scoring rule.

The participation-rate weighted Brier score or variants thereof is used throughout various Cultivate Labs platforms, such as Good Judgment Open, CSET-foretell, or Covid Impacts. For more information, see \cite{GJSR}. Note that Good Judgment may reward the relative ``participation-rate weighted Brier score'', namely: 
\begin{equation}
    RPWBS(p) = \frac{1}{t_2-t_0} \cdot \int_{t_1}^{t_2} Brier(p(t)) - Brier(p_{aggregate}(t))dt
\end{equation}

However, whether the ``participation-rate weighted Brier score'' is relative or not doesn't make a difference to this section and its proof, so we will ignore it\footnote{Note that we can separate the $Brier(p_{aggregate(t)})$ term in its own integral
\begin{equation}
    \begin{split}
        RPWBS(p) &= \frac{1}{t_2-t_0} \cdot \int_{t_1}^{t_2} Brier(p(t)) - Brier(p_{aggregate}(t))dt \\
        &= \frac{1}{t_2-t_0} \cdot \int_{t_1}^{t_2} Brier(p(t)) - \frac{1}{t_2-t_0} \cdot \int_{t_1}^{t_2} Brier(p_{aggregate}(t))dt
    \end{split}
\end{equation}
The $\frac{1}{t_2-t_0} \cdot \int_{t_1}^{t_2} Brier(p_{aggregate}(t))dt$ term is outside the forecaster's control, and thus 
\begin{equation}
    \begin{split}
        RPWBS(p) = \frac{1}{t_2-t_0} \cdot \int_{t_1}^{t_2} Brier(p(t)) - C
    \end{split}
\end{equation}

But minimizing this expression is equivalent to minimizing the expression without the constant, i.e., to maximizing the PWBS. 
} . Additionally, the Cultivate Labs platform, on which Good Judgment Open and other tournaments run, only permits one prediction per day, so the integral would instead be a summation ($\sum$). Again, this doesn't change the conclusion. 

We will provide an overview of our argument, and then a formal proof. As for the overview, the ``participation-rate weighted Brier score'' ceases to be proper when it gets multiplied by the participation rate. The crux of the issue is that if a question set to resolve in 100 days instead resolves earlier, say, after the first 10 days, those who have forecasted during the first 10 days are counted as having 100\% participation, instead of 10\% participation. This then incentivizes higher probabilities at the beginning for questions which are unlikely to be resolved soon, but such that if they resolved soon, it is knowable how they resolve. Many questions have this property, such as whether a prominent figure will leave office, whether a deal or announcement will be made, or whether an event will in general happen. 

\begin{theorem}
    The ``participation-rate weighted Brier score'' is not proper. Further, the misincentive to report a dishonest probability is practically unbounded. 
\end{theorem}

\begin{proof}
    Let ``$PWBS$'' stand for ``participation-rate weighted Brier score''
    Suppose a question has a 10\% chance of resolving positively in two weeks (Scenario 1), and otherwise a 90\% probability of resolving negatively in one year (Scenario 2). Suppose that a forecaster predicts $p_1$ during the first week, and  $p_2=0$ afterwards. Then the expected value of the score is:
    \begin{equation}
        \begin{split}
            E[PWBS(p_1, p_2)] &= 0.1 \cdot PWBS(\textnormal{Scenario 1}) + 0.9 \cdot PWBS(\textnormal{Scenario 2}) \\ 
        \end{split}
    \end{equation}
    \begin{equation}
        PWBS(\textnormal{Scenario 1}) = \frac{1}{t_1-t_0} \cdot \int_{t_0=0}^{t_1=2} (1-p_1)^2 = (1-p_1)^2
    \end{equation}
    \begin{equation}
        \begin{split}
            PWBS(\textnormal{Scenario 2}) = &\frac{1}{t_2-t_0} \cdot \Big( \int_{t_0=0}^{t_1=2} (0-p_1)^2 dt + \int_{t_1 = 2}^{t_2=52} (0-p_2)^2 dt \Big) \\
            &= \frac{2}{52} \cdot (0-p_1)^2 + \frac{50}{52} \cdot (0-p_2)^2 
        \end{split}
    \end{equation}

    \begin{equation}
        \begin{split}
            E[PWBS(p_1, p_2)] &= 0.1 \cdot (1-p_1)^2\\
            &+ 0.9 \cdot \left(\frac{2}{52} \cdot (0-p_1)^2 + \frac{50}{52} \cdot (0-p_2)^2 \right)
        \end{split}
    \end{equation}

    Taking out the $(0-p_2)^2$ term, which is 0, then (with $\frac{2}{52} = 0.03846$):

    \begin{equation}
        E[PWBS(p_1)] = 0.1 \cdot (1-p_1)^2 + 0.9 \cdot 0.03846 \cdot (0-p_1)^2
    \end{equation}

    We can take the derivative of $f(x) = E[PWBS(x)]$ to find the minimum:

    $$f'(p_1) = E'[PWBS(p_1)] = 0.1 \cdot ((-1) \cdot 2\cdot(1-p_1)) + 0.9 \cdot (0.03846 \cdot 2 \cdot p_1)$$

    $$f'(p_1)/2 = -0.1 \cdot (1-p_1)+ 0.9 \cdot (0.03846 \cdot p_1)$$

    $$f'(p_1)/2 = -0.1  + 0.1 \cdot p_1+ 0.034614 \cdot p_1$$

    $$-f'(p_1)/2 = -0.1  + 0.134614 \cdot p_1$$

    Now the minimum is reached when $f'(p_1)=0$, and hence when:
    
    \begin{equation}
        p_1 = \frac{0.1}{0.134614}=0.74286... \approx 74\%   
    \end{equation}

    which is a significant misincentive (i.e., declaring a 74\% when one really believes 10\%). 
    
    In general, per reasoning analogous to the above, if the true probability is $t$ and the first interval (e.g. 0.034614 for 2 out of 52 weeks in a year) is $i$, then the optimal probability to declare is  $\frac{t}{t+i}$, \textit{making the error unbounded}, that is, if $i$ is small enough, then one might be incentivized to input $100\%$ when one's true probability is $1\%$.

    In practice, Good Judgement Open and other Cultivate Labs platforms rarely have questions longer than two years, and only the last update for each forecaster is counted for each day. So if $t=1\%$, and $i=$ one day in two years, then one would be incentivized to input: 
    
    $$ p_1 = \frac{1}{1+(1/(2 \cdot 365))} = 99.86\%$$

    Note that, because Cultivate Labs platforms only accept integers as probabilities, the above might be rounded up to $100\%$ by forecasters. This means that in the case above, forecasters are incentivized to enter a probability of 100\% even when they believe the true probability to be 1\%
\end{proof}

One might object that this type of misincentive does not occur in most questions, and that we are merely counting beans. However, if at the beginning of the Cuban missile crisis Good Judgment had asked about the probability of a nuclear exchange in the next two years, this category of error might have applied in full force, yet it might have been all the more valuable to have access to highly calibrated probabilities. 

Empirically, questions such as \href{https://www.gjopen.com/questions/1820-will-prayut-chan-o-cha-cease-to-be-the-prime-minister-of-thailand-before-23-april-2021}{Will Prayut Chan-o-cha cease to be the prime minister of Thailand before 23 April 2021?}, display behavior broadly consistent with what optimal (dishonest) behavior would recommend. That is, the probabilities for the question are excessively high at the beginning. 

\begin{figure}[h!]
    \includegraphics[width=12cm]{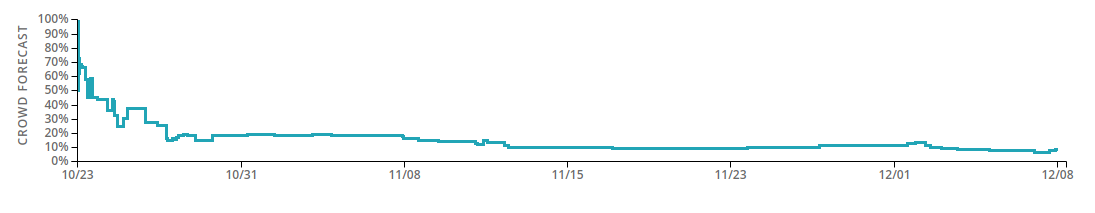}
    \centering
    \caption{Will Prayut Chan-o-cha cease to be the prime minister of Thailand before 23 April 2021?: Consensus probability over time on Good Judgement Open}
\end{figure}

In contrast, updating either according to Laplace's rule, or just naïvely updating on the passage of time\footnote{I.e., where the probability starts at $0.66=66\%$, and the remaining probability at timetime $t$ is given by $$p(t)=1-(1-0.66)^{\frac{182-t}{182}}$$} produces very different update shapes, per Figures 6 and 7.

\begin{figure}[h]
    \includegraphics[width=12cm]{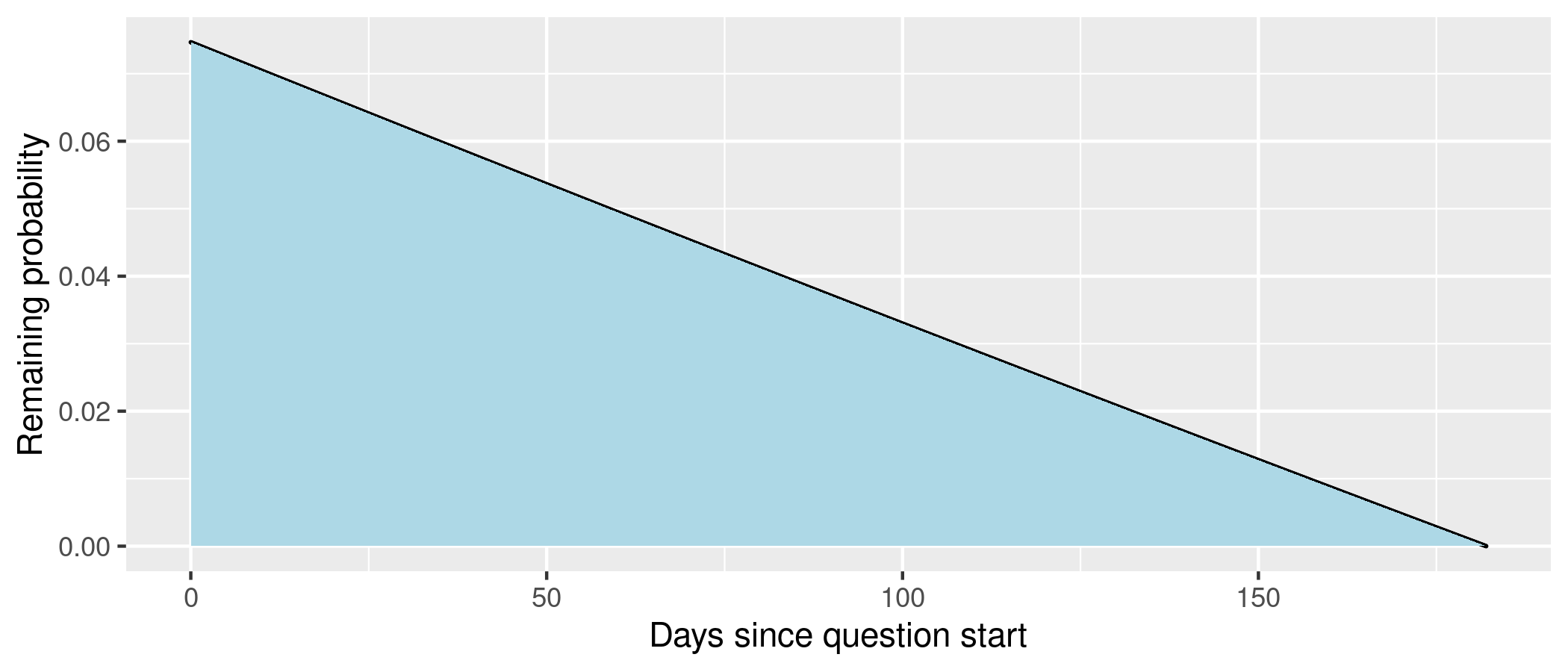}
    \centering
    \caption{Decreasing probability per Laplace's rule of Prayut Prayut Chan-o-cha ceasing to be the Prime Minister by April 23rd 2021, who at question start had been Prime Minister for 2346 days}
\end{figure}

\begin{figure}[h]
    \includegraphics[width=12cm]{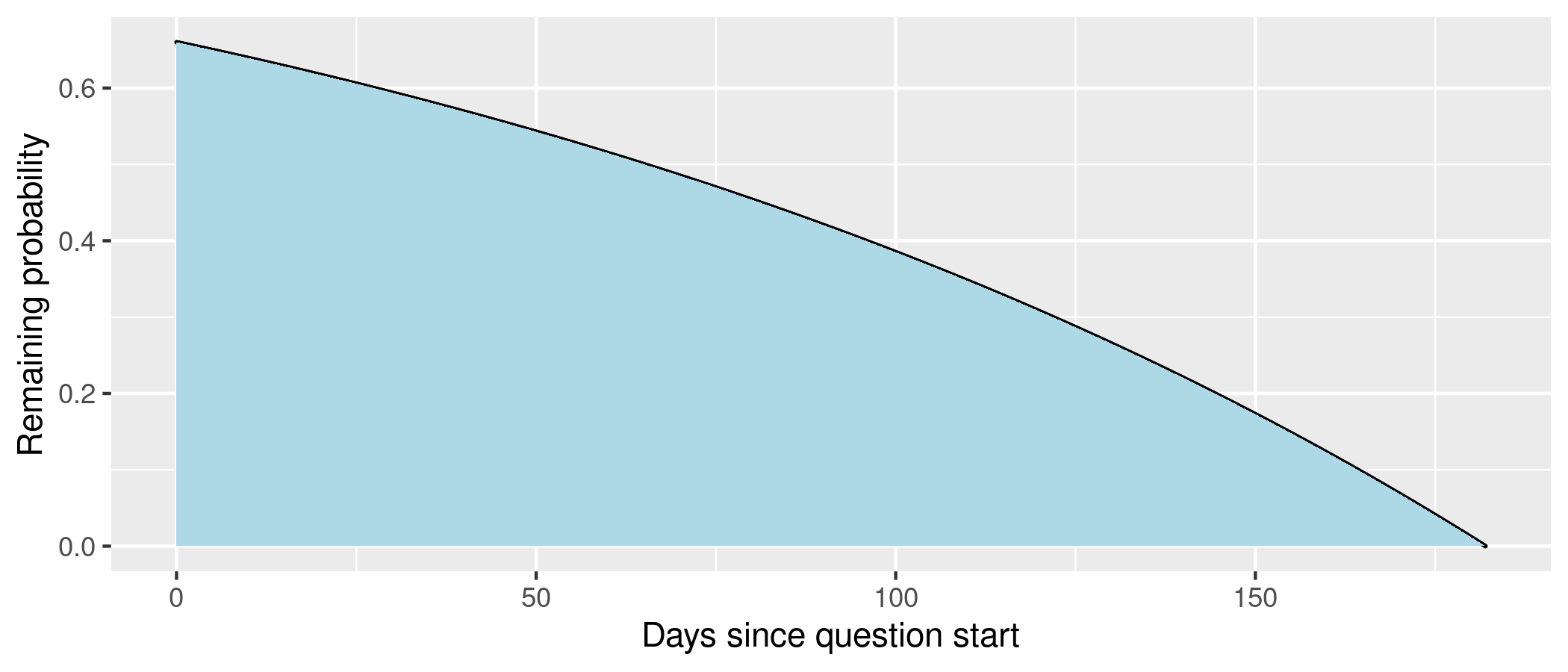}
    \centering
    \caption{Decreasing probability per ``naïve updating" of Prayut Prayut Chan-o-cha ceasing to be the Prime Minister by April 23rd 2021, who at question start had been Prime Minister for 2346 days}
\end{figure}

We do not believe that this is a result of conscious decisions on the part of forecasters. However, positing such conscious decisions isn't necessary; see \cite{lagerros} for an accessible discussion about ``Unconscious Economics'', as well as \cite{friedman2} and  \cite{friedman1} for a more in-depth treatment.

\subsubsection{Implications for "extremization"}

"Extremization" refers to taking aggregate probabilities and moving them towards 0\% or 100\%. Quoting \cite{Tetlock}, pp.90-91:
\begin{quote}
    Then you give the forecast a final tweak: You ``extremize" it, meaning you push it closer to 100\% or to zero. If the forecast is 70\%, you might bump it up to, say, 85\%. If it's 30\%, you might reduce it to 15\%. \\
    Now imagine that the forecasts you produce this way beat those of every other group and method available, often by large margins. Your forecasts even \textbf{beat} those of professional intelligence analysts inside the government who have access to classified information--by margins that remain classified. [...] It actually happened. What I've described is the method we used to win IARPA's tournament. [Emphasis mine]
\end{quote}

In this section, we prove that there are cases where extremization of otherwise accurate probabilities is incentivized by the ``participation-rate weighted Brier score", which is improper, and wouldn't be incentivized by a proper scoring rule. We will do this by constructing two scenario in which this is the case. In the first scenario, the effect to distort probabilities is small. In the second scenario, the effect is larger.

\begin{theorem}
    The ``participation-rate weighted Brier score'' rewards dishonest extremization. 
\end{theorem}
\begin{proof}
Consider an event which for the next 5 days, has a 25\% probability of happening at noon each day, provided it hasn't happened before. Then the probability of it happening at least once, given that it hasn't happened before, is:

\begin{center}
\begin{tabular}{ |c|c|c|c| } 
 \hline
  Point & Point in time & Probability & Expression \\ 
 \hline
 \hline
 $S_1$ & Before the 1st day & .68359375 & $1-(1-0.25)^4$ \\ 
 $S_2$ & After the 1st day & .578125 & $1-(1-0.25)^3$ \\ 
 $S_3$ & After the 2nd day & .4375 & $1-(1-0.25)^2$ \\ 
 $S_4$ & After the 3rd day & 0.25 & $1-(1-0.25)^1$ \\ 
 $S5$ & After the 4th day & 0 & $1-(1-0.25)^0$ \\ 
 \hline
\end{tabular}
\end{center}

Let's ``extremize" these probabilities

\begin{center}
\begin{tabular}{ |c|c|c| } 
 \hline
 Point & Point in time & Extremized probabilities \\
 \hline
 \hline
 $S_1$ & Before the 1st day & .75 \\
 $S_2$ & After the 1st day & .65 \\
 $S_3$ & After the 2nd day & .35  \\
 $S_4$ & After the 3rd day & 0.2 \\
 $S_5$ & After the 4th day & 0 \\
 \hline
\end{tabular}
\end{center}

Per calculations in appendix \S\ref{Extremization calculations}, the expected ``participation-rate weighted Brier score'', E[PWBS] is, respectively:

\begin{equation}
    \begin{split}
        E[PWBS(\textnormal{honest probabilities)}] = 0.193 \\
        E[PWBS(\textnormal{extremized probabilities)}] = 0.182
    \end{split}
\end{equation}

Conversely, the order is reversed if we use a proper version of the PWBS, which we define in \S\ref{Do not truncate}

\begin{equation}
    \begin{split}
        E[\textit{Proper } PWBS(\textnormal{honest probabilities)}] = 0.154 \\
        E[\textit{Proper } PWBS(\textnormal{extremized probabilities)}] = 0.160
    \end{split}
\end{equation}

\end{proof}

We note, that the effect in the example is small, on the order of less than $10\%$ of the score\footnote{In this example, the extremizer gets a $(0.193-0.182)/0.182=0.06=6\%$ advantage when it should have been getting a $(0.160-0.154)/0.16 = 0.0375=3.75\%$ disadvantage}. And this seems to be the case for questions in which the forecast decreases continuously with time. 

Consider a second scenario, in which the probability of an event is 60\% during the first half of a question, and 0\% afterwards. Then the expected participation-weighted Brier score for predicting a probability of $x$ during the first half and $0$ during the second half is given by

\begin{equation}
    PWBS(x) = 0.6 \cdot (1-x)^2 + 0.4 \cdot \left(\frac{1}{2} \cdot x^2 + \frac{1}{2} \cdot 0\right)
\end{equation}

The minimum is given when $x=0.75 = 75\%$, where the PWBS is $0.15$, in contrast with an expected PWBS of $0.168$ when revealing the true probability of $60\%$. If the scoring rule had been proper, as in, a 75\% probability would get a score of $0.2625$, whereas a 60\% probability would have gotten a brier score of $0.24$. This is a slightly larger effect, on the order of 20\% of the score\footnote{$PWBS(0.75) = 0.6 \cdot (1-0.75)^2 + 0.4 \cdot 0.5 \cdot 0.75^2
= 0.15$. $PWBS(0.6) = 0.6 \cdot (1-0.6)^2 + 0.4 \cdot 0.5 \cdot0.6 ^2 = 0.168$. 
The difference is $(0.168-0.15)/0.168 = 0.1075 = 10.75\%$
If the participation Brier score had been proper, $BS(0.6) = 0.6*(1-0.6)^2 + 0.4*(0-0.6)^2 = 0.24$, $BS(0.75) = 0.6*(1-0.75)^2 + 0.4*(0-0.75)^2 = 0.2625$, and the difference is $(0.2625-0.24)/0.24 = 0.0225 = 0.09375 = 9.375\%$. So the extremizer got a $10.75\%$ advantage when it should have been getting a $9.375\%$ disadvantage, for a total difference of $9.375 + 10.75 = 20.125\%$. Percentage differences of the Brier score seems like a fairly arbitrary unit.}.

At this point, we have two competing considerations. On the one hand, per \cite{gjfirst}, the median of Good Judgment Project superforecasters' forecasts in the original IARPA ACE competiton ``was 35\% to 72\% more accurate than any other research team in the competition''. This suggests that, even if the original Good Judgment Project gained an illegitimate advantage of $10\%$ to $20\%$ over competitors by identifying and exploiting an improper scoring rule, it would still have won otherwise. 

Further, per \cite{gjscience}, the Good Judgment Project's advantage is due to multiple factors: Selection of forecasters with suitable mental traits, de-biasing training, teaming, and aggregation algorithms.

\begin{figure}[h]
    \includegraphics[width=12cm]{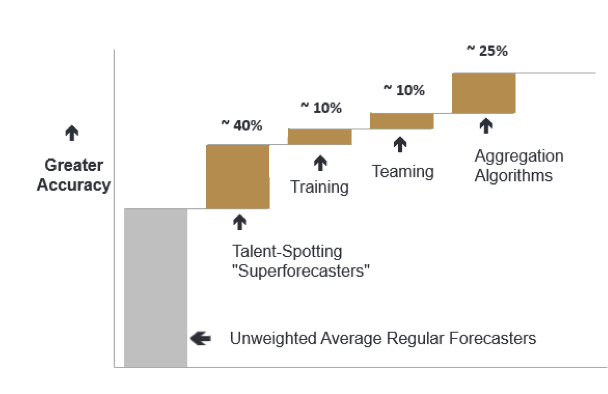}
    \centering
    \caption{``Talent-Spotting, Training, Teaming and Aggregation: Our Evidence-Based Formula for Enhanced Accuracy''. Taken from \cite{gjscience}.}
\end{figure}

This also suggests that, in the aggregate, the techniques and skills identified by the Good Judgment Project and their superforecasters would still be valuable even if they are 10\% or 20\% worse than previously considered. 

On the other hand, given strong optimization pressure, and in particular data analysis to find optimal aggregation algorithms to optimize an improper scoring rule, we would expect many patterns to be found and exploited. And this would include patterns that take advantage of the difference between the proper and improper scoring rules. Or, in other words, patterns which improve the score under the improper scoring rule no matter at what expense to true accuracy. For instance, in our last example scenario, we found that manipulating the probability from 60\% to 75\% would improve the PWBS score by 12\%. This is a small improvement in the score which comes at the expense of a relatively large manipulation in base percentage points. 


We contacted Good Judgment Inc, and they revealed that, per a previously unpublished paper, (\cite{Karvetski}), they were already aware of the incentive problems with participation weighted Brier scores. They pointed this out to IARPA during the Aggregative Contingent Estimation (ACE) Program, but other participants voted not to correct the flaw. Further, the in-house forecasting platform which Superforecasters use does have a proper scoring rule. 

\subsection{Reputational rewards may outweigh internal platform rewards}
If the signalling value a predictor obtains from a high rank outweighs the difference in expected monetary reward between strategies, then probabilistic rewards will not prevent the incentive to extremize in order to be ranked as the "best predictor" in a competition. As many competitions run on many platforms, the risk of performing badly in a competition is unlikely to solve this problem alone, as predictors can select which results they publicise.


\section{Solutions}
\label{Solutions}

\subsection{Use probabilistic rewards}
\label{Probabilistic rewards}

In tournaments with a fairly small numbers of questions, where paying only the top few places would incentivize forecasters to make overconfident
predictions to maximize their chance of first-place finishes as discussed above, probabilistic rewards may be used to mitigate this effect.

In this case, rather than e.g. having prizes for the top three scorers, prizes would be distributed  according to a lottery, where the number of tickets each player received was some function of their score, thereby incentivizing players to maximize their expected score, rather than their chance of scoring highest.

Which precise function should be used is a non-trivial question: If the payout structure is too ``flat'' i.e. making a much better prediction than average does not get you many more tickets, then there is not sufficient incentive for people to work hard on their forecasts compared to just entering with the community median or some reasonable prior. If on the other hand the payout structure too heavily rewards beating the crowd forecast, then the original problem returns.

\subsection{Give rewards to the best forecasters among many questions.}
\label{Many questions}

If one gives rewards to the top three forecasters for 10 questions in a contest in which there are 200 forecasters, the top three forecasters might be as much a function of luck as of skill, which might be perceived as unfair. In fact, as demonstrated above, there are situations where more skill (in the sense of being closer to the "true" probability) is associated with a \emph{lower} chance of success (being in the top few.) As we discussed in previous sections, giving prizes for much larger pools of questions makes this effect smaller, allowing for skill to dominate luck of the draw effects. 

\subsection{Force forecasters to forecast on all questions}
\label{Forcing}

This fixes the incentive to pick easier questions to forecast on. A similar idea would be to assume that forecasters have predicted the community median on any question that they haven't forecast on, until they make their own prediction, and then reporting the average brier score over all questions. This has the disadvantage of not rewarding ``first movers/market makers", however it has the advantage of ``punishing" people for not correcting a bad community median in a way that the relative Brier score doesn't. It also removes the incentive to manually forecast on all questions.

\subsection{Score groups}
\label{Scoring groups}

If one selects a group of forecasters and offers to reward them in proportion to the Brier score of the group's predictions for a fixed set
of questions, then the forecasters now have the incentive to share information with the group. This group of forecasters could be pre-selected for having made good predictions in the past.

\subsection{Anonymize winners}
If reputational rewards from placing in the top few outweigh monetary rewards, and forecasters are able to publicize parts of their track record, the incentive to distort forecasts in order to be in the top few returns. This is not a problem if winners are not known. On the other hand, monetary rewards might have to be increased in the absence of reputational rewards.

\subsection{Divide information gatherers and prediction producers}
\label{Divide info gatherers}

If this is done, information gatherers might then be upvoted by prediction producers, who would have less of a disincentive not to do so, though hiding good comments from other forecasters would still incentivised. Alternatively, some prediction producers might be shown information from different information gatherers, or select which information was responsible for a particular change in their forecast. A scheme in which the two tasks are separated might also lead to efficiency gains.

\subsection{Do not truncate the participation rate}
\label{Do not truncate}
If one does not truncate the participation rate, Good Judgement Open's score becomes proper. That is, if a question resolves when it is $n\%$ towards completion, the forecasters' Brier score should be multiplied by $n\%$, not by $100\%$. Formally, if a question would otherwise have been left open until $t_3$, the weighting should be as

\begin{equation}\label{properalternative}
    \frac{1}{t_3-t_0} \cdot \int_{t_1}^{t_2} Brier(P(t)) dt
\end{equation}

\begin{proofsketch}
The integral is the limit of having a question for each period of time $h$, each of which is weighted by $\frac{h}{t_3-t_0}$, and scoring each separately. That is, the integral is similar to having a question for each day and that question being scored with a proper Brier score, a setup which would be proper.

But because $g(t) = Brier(P(t))$ is Riemann integrable, the limit of the sums is as we expect:

\begin{equation}
    \lim_{n\to\infty} \sum_{k=0}^{n} g\left(t_1 + \frac{k}{n} \cdot t_2\right) = \int_{t_1}^{t_2} g(t) dt
\end{equation}

Thus, if a series of P(t) minimize the discrete case, they also minimize the continuous case. Because the minimum of the discrete case is achieved when the forecasters honestly reports their probabilities, that will be the minimum for the continuous case as well. The scoring rule in (\ref{properalternative}) is thus proper. 
\end{proofsketch}

One consequence of not truncating the participation rate is that questions which resolve early end up being worth less to the user. This can cause confusion, especially with additive scoring rules such as those used on the Metaculus platform, which correctly does not truncate their scores. New users express surprise frequently enough at participation rate truncation that an FAQ entry on this exact topic was recently added on the \cite{FAQ} platform.

\subsection{Design collaborative scoring rules}
\label{Collaborative scoring rules}

Tournament organizers could seek to use collaborative scoring rules, where forecasters are incentivized to share information with each other. We define a proper collaborative scoring rule as one where the highest expected reward is obtained if a user maximises both their own score and the aggregate score. If forecasting tournament organizers are able to produce a prior, and $Brier$ denotes the Brier score, the following is a proper collaborative scoring rule:

\begin{equation}
    S = \alpha \cdot Brier(\textnormal{forecaster}) + \beta \cdot Brier(\textnormal{aggregate}) - \gamma \cdot Brier(\textnormal{prior})
\end{equation}

The proof is left to the reader. One might also seek to design other collaborative scoring rules which attempt to approximate the Shapley value of each forecaster. Note that for this scoring rule to be proper, S must directly be the reward the forecaster obtains. If S is used as a score, with that score determining how some finite prize is shared, the incentive to collaborate disappears. Unfortunately, making S directly be the reward means that the size of the prizepool for a tournament cannot be decided in advance, which makes this system somewhat unattractive for funders.

Currently, some platforms make it possible to give upvotes to the most insightful forecasters, but if upvotes were monetarily rewarded, one might not have the incentive to upvote other participants' contributions as opposed to waiting for one's contributions to be upvoted. Insightful comments can also be externally judged and incentivised. A reward system the authors intend to field test in the near future involves comments that are judged to be useful being rewarded with a multiplier on question scores, therefore ensuring that helpful sounding comments are not rewarded if accompanied by poor forecasts.

In practice, Metaculus and Good Judgment Open do have healthy communities which collaborate, where trying to maximize accuracy at the expense of other forecasters is frowned upon, but this might change with time, and it might not always be replicable. In the specific case of Metaculus, monetary prizes are relatively new, but becoming more frequent. It remains to be seen whether this will change the community dynamic.

\section{Conclusions}
\label{Conclusion}

We presented some incentive problems which currently befall various forecasting platforms and contests. These incentive problems are worrying because the probabilities which these platforms and contests produced generally aim to be action-guiding, but forecasters have an incentive to input distorted probabilities. Following \cite{lagerros}, these incentives can either manifest themselves because forecasters consciously choose to maximize them, because the forecasters who happen to follow them rise to the top, or because forecasters internalize them in their learning processes when interpreting flawed scoring mechanisms as feedback.

Further, improper scoring rules are then used to grade and reward forecasters. In the case of Good Judgment Open, superforecasters are selected amongst the best performers, but this is done according to a distorted, improper scoring rule. We also find that ``extremizing"---pulling aggregate probabilities towards the nearest of 0\% and 100\%---is in some cases incentivized by Good Judgment Open's improper scoring rule, and thus some of Tetlock et al.'s results about extremization might not hold if they were discovered while using a proper forecasting rule. However, we are unsure about the size of this effect, and subjectively estimate it to be small.

We also pointed out that common scoring rules incentivize forecasters not to share information, and that that tournaments with discrete prizes generate an incentive to distort forecasts to maximize the chances of winning, rather than to maximize expected loss.


For the case of aligning human forecasters, solving the alignment problem would have required using a truly proper scoring rule, making it incentivize cooperation, etc. But the most prominent forecasting platforms and competitions have, to a certain extent, failed to do so. We can relate this problems with the broader alignment problem present in machine learning systems. In particular, we notice that the incentive problems in our forecasting systems resemble machine learning specification gaming examples, such as those outlined in \cite{Krakovna}, whereas others are more like principal-agent problems, which bear some resemblance to the inner alignment problems in \cite{Hubinger}.

Perhaps incentive design problems in forecasting systems could be solved or mitigated using tools from machine learning or artificial intelligence alignment. Conversely, perhaps incentivizing creative human forecasters seeking prestige and monetary rewards to produce useful forecasts might serve as a toy problem or test ground for machine learning researchers seeking to align potentially very powerful artificial systems. Lastly, the fact that the ``forecasting alignment problem'' hasn't been solved might provide an opportunity to test our alignment capabilities in preparation for the harder problem of aligning hypothetical future machines more intelligent than humans.




\appendix








\begin{thebibliography}{00}

\bibitem[CSET-Foretell (2020)]{SouthChina}
    CSET-Foretell (2020)
    ``Will the Chinese military or other maritime security forces fire upon another country's civil or military vessel in the South China Sea between January 1 and June 30, 2021, inclusive?''
    URL: https://web.archive.org/web/20201031221709/https://goodjudgment.io/ superforecasts/\#1338

\bibitem[Enten (2017)]{Fake polls real problem}
    Enten, H. (2017)
    ``Fake Polls Are A Real Problem"
    URL: https://fivethirtyeight.com/features/fake-polls-are-a-real-problem/

\bibitem[Friedman (2001)]{friedman1}
    Friedman, D. (2001)
   \textit{Law's Order: What Economics Has to Do with Law and Why It Matters}

\bibitem[Friedman (1990)]{friedman2}
    Friedman, D. (1990)
    \textit{Price Theory: An Intermediate Text}.

\bibitem[Good Judgement (2018)]{gjscience}
    Good Judgment (2018)
    ``The Science of Superforecasting''
    Archived URL: https://web.archive.org/web/20180408044422/http://goodjudgment.com/science.html
    
\bibitem[Good Judgment Scoring Rule (2019)]{GJSR}
    Good Judgment (2019) 
    ``4. How are my forecasts scored for accuracy?''
    URL: https://www.gjopen.com/faq\#faq4

\bibitem[Good Judgement (2020a)]{covidrecovery}
    Good Judgment (2020a)
    ``Public Dashboard''
    URL: https://goodjudgment.com/covidrecovery/
    archived URL: https://web.archive.org/web/20201120231552/ https://goodjudgment.com/covidrecovery/

\bibitem[Good Judgement (2020b)]{USelections}
    Good Judgment (2020b)
    ``Who will win the 2020 United States presidential election?''
    URL: https://web.archive.org/web/20201031221709/ https://goodjudgment.io/superforecasts/\#1338

\bibitem[Good Judgement (2020c)]{gjfirst}
    Good Judgment (2020c)
    ``The First Championship Season''
    Archived URL:
    https://web.archive.org/web/20201127110425/https://goodjudgment
    .com/resources/the-superforecasters-track-record/the-first-championship-season

\bibitem[Hubinger et al. (2019)]{Hubinger}
    Hubinger, E. et al. (2019)    
    ``Risks from Learned Optimization in Advanced Machine Learning Systems"
    URL: https://arxiv.org/abs/1906.01820

\bibitem[Krakovna et al. (2020)]{Krakovna}
    Krakovna V. et al. (2020)
    ``Specification gaming: the flip side of AI ingenuity"
    URL: https://deepmind.com/blog/article/Specification-gaming-the-flip-side-of-AI-ingenuity/
    Specification gaming examples in AI - master list: https://docs.google.com/spreadsheets/d/e/2PACX-1vRPiprOaC3HsCf5Tuum8bRfzYUiKLRqJmbOoC-32JorNdfyTiRRsR7Ea5eWtvsWzuxo8bjOxCG84dAg/pubhtml
    
\bibitem[Karvetski et al. (s.a)]{Karvetski}
    Karvetski C.,  Minto T., Twardy, C.R. (s.a.)
    ``Proper scoring of contingent stopping questions".
    Unpublished.
    
\bibitem[Lagerros (2019)]{lagerros}
    Lagerros, J. (2019)
    ``Unconscious Economics''
    URL: https://www.lesswrong.com/posts/PrCmeuBPC4XLDQz8C/unconscious-economics

\bibitem[Lichtendahl et al. (2007)]{Lichtendahl}
    Lichtendahl et al.
    ``Probability Elicitation, Scoring Rules, and Competition Among Forecasters''
    Management Science, Vol. 53. N. 11. 
    URL: https://pubsonline.informs.org/doi/abs/10.1287/mnsc.1070.0729?journalCode=mnsc
    
\bibitem[Metaculus (2018)]{ragnarok}
    Metaculus (2018)
    ``Ragnarök Question Series''
    URL: https://www.metaculus.com/questions/1506/ragnar\%25C3\%25B6k-question-series-overview-and-upcoming-questions/

\bibitem[Metaculus (2020)]{aiprogress}
    Metaculus (2020)
    ``Forecasting AI Progress''
    URL: https://www.metaculus.com/questions/1506/ragnar\%25C3\%25B6k-question-series-overview-and-upcoming-questions/
    
\bibitem[Metaculus (2021)]{FAQ}
    Metaculus (2021)
    ``FAQ''
    URL: https://www.metaculus.com/help/faq/\#fewpoints

\bibitem[Rodriguez (2019)]{Rodriguez}
    Rodriguez, L. (2019) 
    ``How many people would be killed as a direct result of a US-Russia nuclear exchange?''
    URL: https://forum.effectivealtruism.org/posts/FfxrwBdBDCg9YTh69/how-many-people-would-be-killed-as-a-direct-result-of-a-us

\bibitem[Tetlock et al. (2015)]{Tetlock} 
    Tetlock, P. \& Gardner, D. (2015) \textit{Superforecasting: The Art and Science of Prediction.} 

\bibitem[Witkowski et al. (2021)]{Witkowski}
    Witkowski J. et al. (2021) ``Incentive-Compatible Forecasting Competitions"
    URL: https://arxiv.org/abs/2101.01816v1

\bibitem[Yeargain (2017)]{Yeargain}
    Yeargain, T. (2020)
    Missouri Law Review, Vol. 85, Issue 1.
    ``Fake Polls, Real Consequences: The Rise of Fake Polls and the Case for Criminal Liability Case for Criminal Liability" (pp. 140-150).
    URL: https://scholarship.law.missouri.edu/cgi/viewcontent.cgi?article=4418
    \&context=mlr 
\end{thebibliography}

\bibliographystyle{plain}

\newpage
\begin{appendices}

\section{Numerical Simulations}\label{Simulations}

\subsection{Method used in the main body of the paper}
To quantify the optimal amount of distortion, we simulate a tournament many times, and observe the results. A tournament is made out of questions and users. 

We model questions as logistic distributions, with a mean of 0, and a standard deviation itself drawn from a logistic distribution of mean 20 and standard deviation 2. For instance, a question might be a logistic distribution of mean 0 and standard deviation 15. At question resolution time, a point is randomly drawn from the logistic distribution. The code to represent this looks roughly as follows:

\begin{verbatim}
generateQuestion = function(meanOfTheStandardDeviation,
    standardDeviationOfTheMean){
  mean <- 0
  sd <- randomDrawFromLogistic(meanOfTheStandardDeviation,
    standardDeviationOfTheMean)
  questionResult <- randomDrawFromLogistic(mean, sd)
  question <- c(mean, sd, questionResult)
  return(question)
}
\end{verbatim}

Users attempt to guess the mean and standard deviation of each question, and each guess has some error. The code to represent this looks roughly as follows:

\begin{verbatim}
generateUser = function(meanOfTheMean, standardErrorOfTheMean,
    meanOfTheStandardDeviation, standardErrorOfTheStandardDeviation){
  user <- function(question){
    questionMean <- question[1]
    questionSd <- question[2]
    questionResolution <- question[3]
    questionMeanGuessedByUser <- questionMean + 
        randomDrawFromLogistic(meanOfTheMean, standardErrorOfTheMean)
    questionSdGuessedByUser <- questionSd +
        randomDrawFromLogistic(meanOfTheStandardDeviation, standardErrorOfTheStandardDeviation))
    probabilityDensityOfResolutionGuessedByUser <-
        getLogisticDensityAtPoint(questionResolution, 
            questionMeanGuessedByUser, questionSdGuessedByUser)
    return(probabilityDensityOfResolutionGuessedByUser)
  }
  return(user)
}
\end{verbatim}

We model the average user as having 
\begin{itemize}
    \item \verb|meanOfTheMean=5|
    \item \verb|standardErrorOfTheMean=5|
    \item \verb|meanOfTheStandardDeviation=5|
    \item\verb|standardErrorOfTheStandardDeviation=5|. 
\end{itemize}

We then consider a "perfect predictor"---a user who knows what the mean and the standard deviation of a question are---and consider how much that perfect predictor would want to distort her own guess to maximize her chances of placing in the top 3 or users. More details can be found in the \href{https://github.com/NunoSempere/Online-Appendix-to-Incentive-Problems-In-Forecasting-Tournaments}{Online Appendix} accompanying this paper.

\subsection{Simluations with more complex distributions of players}
\subsubsection{For binary questions}
A binary question elicits a probability from 0 to 100\%, and is resolved as either true or false. They exist in all three platforms we consider (Metaculus, Good Judgment Open and CSET). 

For the binary case, we first consider a simulated tournament with 10 questions, each with a `true' binary probability between 0 and 100\%. We also consider the following types of forecasters:

\begin{enumerate}
    \item Highly skilled predictors: 10 predictors predict a single probability on each question. Their predictions are off from the "true" binary probability by anywhere from -0.5 to 0.5 bits
    \item Unsophisticated extremizers: 10 highly skilled predictors (whose predictions are off from the "true" binary probability by anywhere from -0.5 to 0.5 bits) who extremize their probabilities by 0.3 bits
    \item Sophisticated extremizers: 5 highly skilled predictors (whose predictions are off from the "true" binary probability by anywhere from -0.5 to 0.5 bits) who take the question closest to 50\% and randomly move it to either 0\% or 100\%
    \item Unskilled predictors: 10 predictors predict a single probability on each question. Their predictions are off from the "true" binary probability by anywhere from -2 to 2 bits
\end{enumerate}

We ran the tournament 10,000 times. We find that sophisticated extremizers do best, followed by lucky unskilled predictors, followed by unsophisticated extremizers, followed by honest highly skilled predictors. 

\newpage

\begin{figure}[h!]
    \includegraphics[width=10cm]{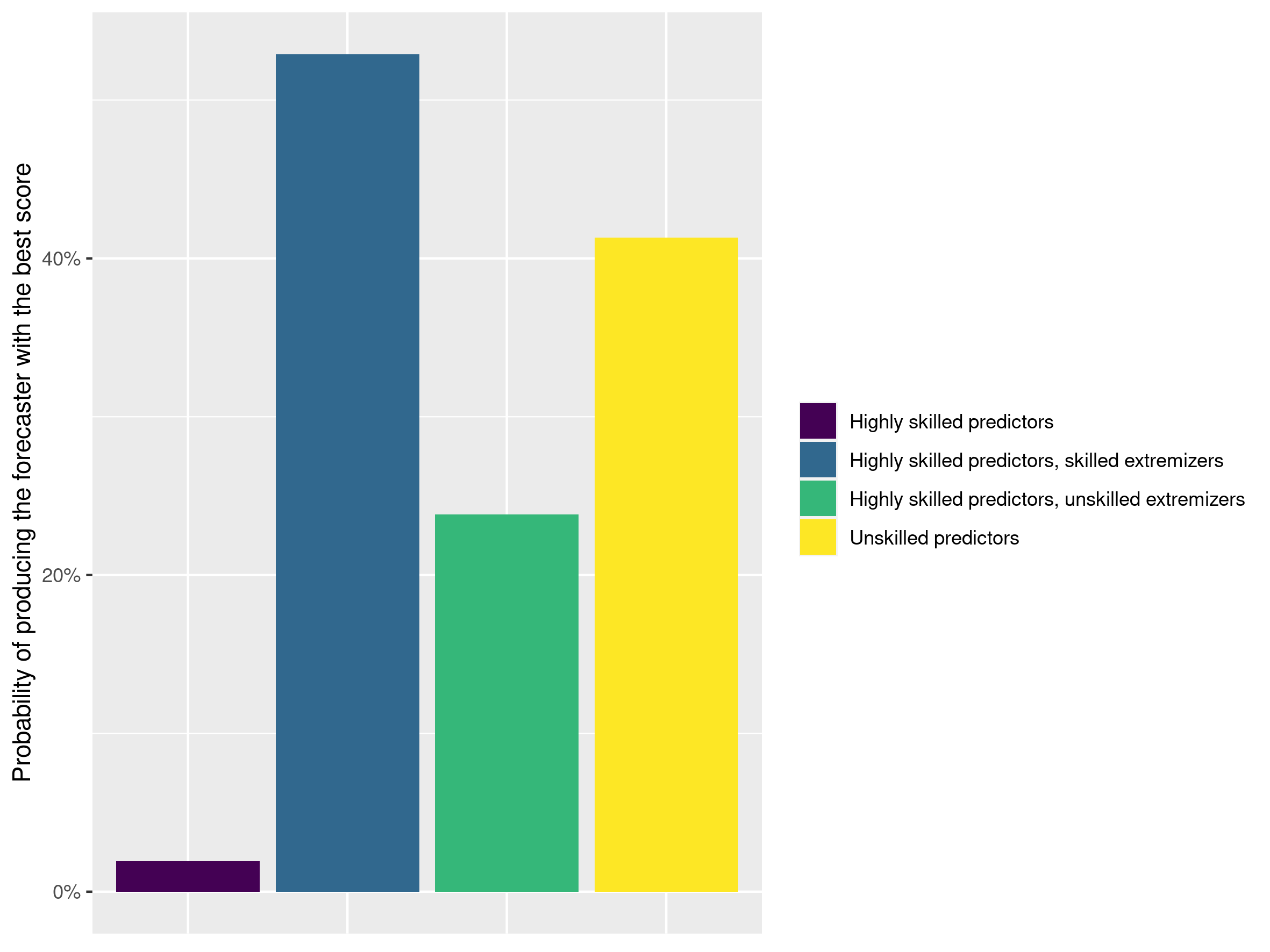}
    \centering
    \caption{\% of the time different predictors reach the top 5}
\end{figure}

\begin{figure}[h!]
    \includegraphics[width=10cm]{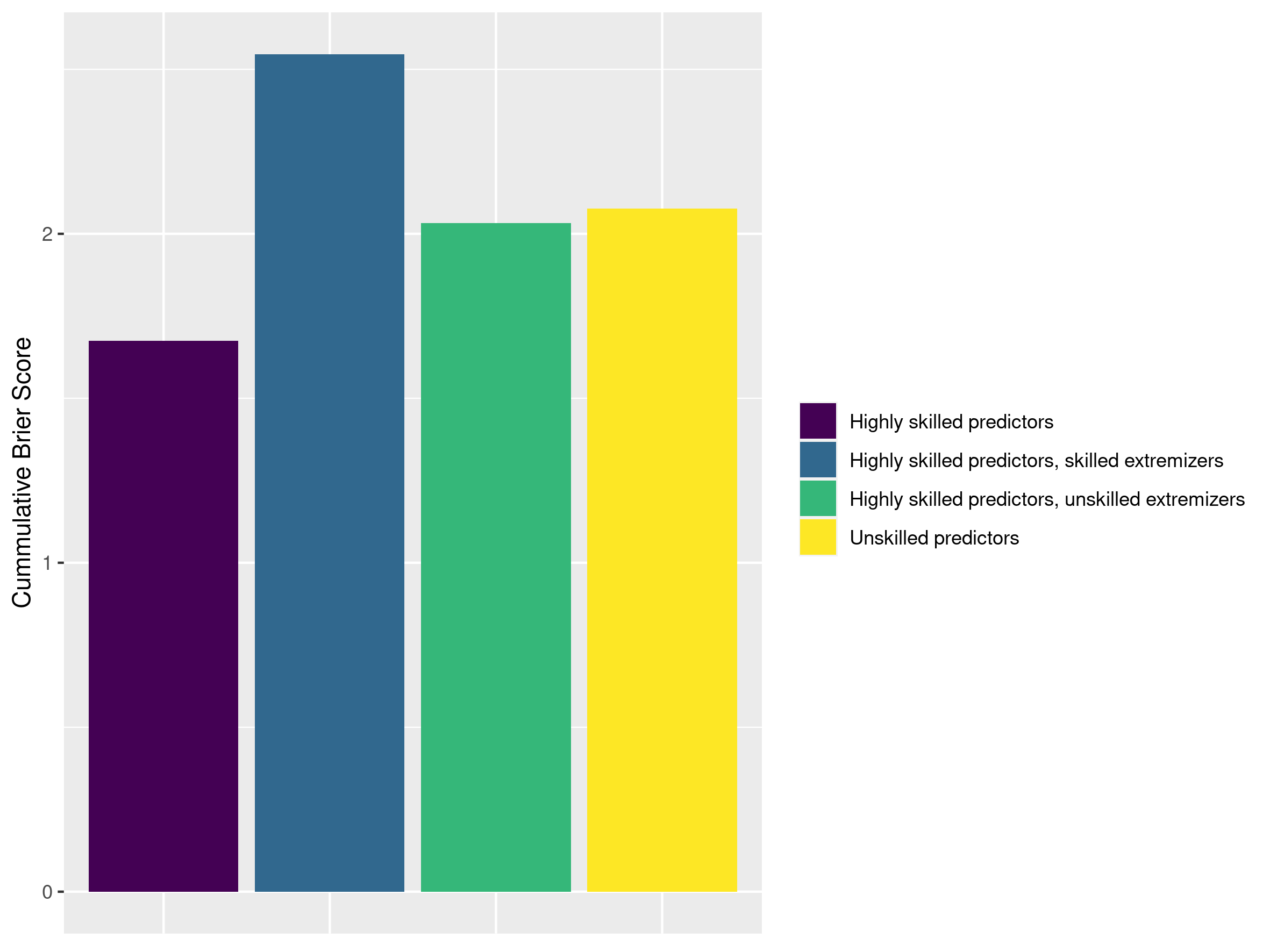}
    \centering
    \caption{Mean Brier score for each group (lower is better)}
\end{figure}
\newpage

We can also explore how this changes with the number of questions. In this case, we ran 1000 simulations for each possible number of questions, in increments of five. Results were as follows:
\\
\begin{figure}[h!]
    \includegraphics[width=\textwidth]{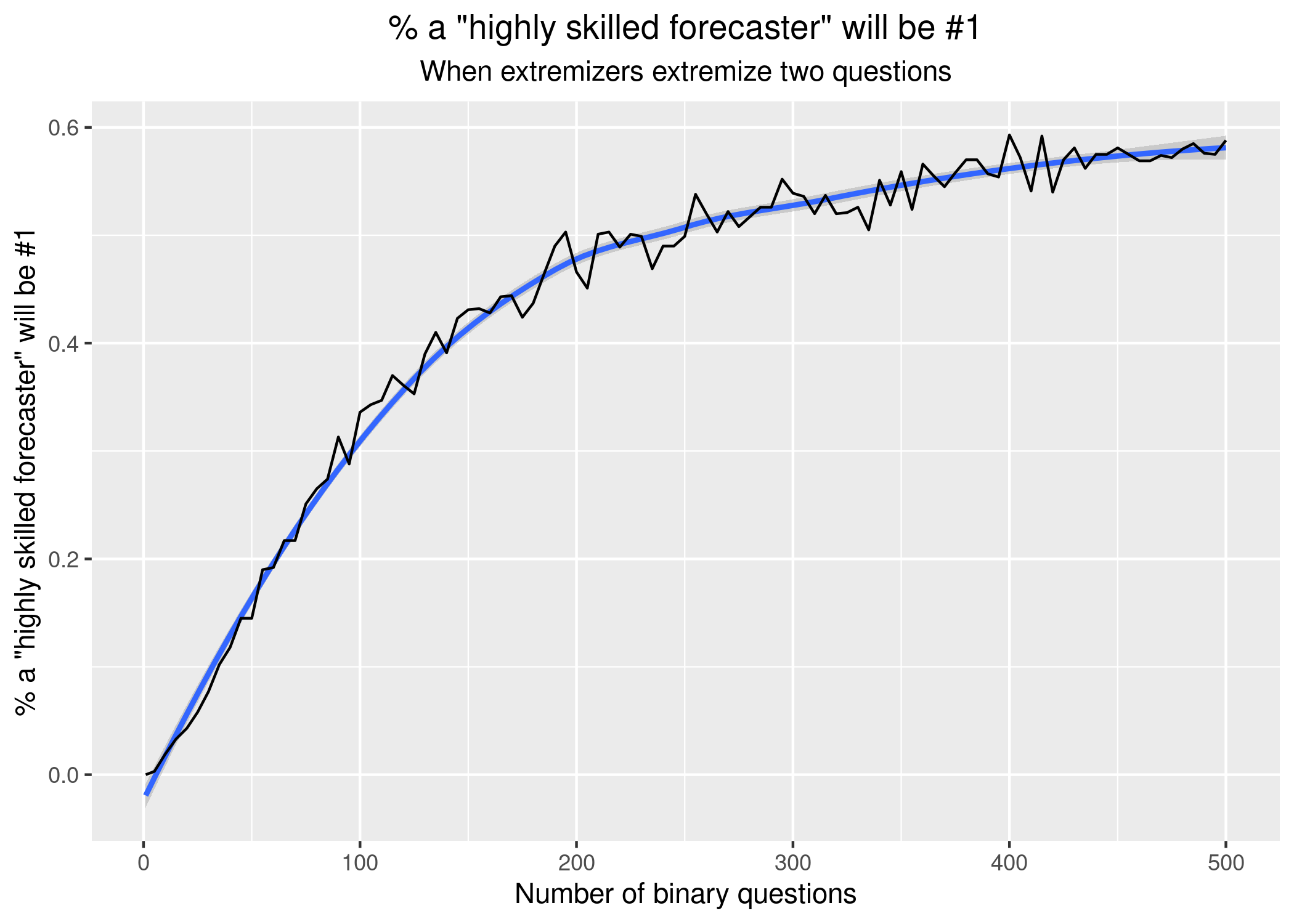}
    \centering
    \caption{Probability that a ``highly skilled predictor'' will obtain the lowest Brier score, for a tournament with $n$ binary questions}
\end{figure}

The results would further vary in terms of how many questions the forecasters who extremize choose to extremize. Dynamic selection of number of questions to extremize based on the total number of questions could provide further opportunities for improved exploitation, but this was not explored.  
\newpage
\subsubsection{For continuous questions}

As before, a continuous question elicits a probability distribution, and is resolved as a resolution value, where the different participants are rewarded in proportion to the probability density they assigned to that resolution value. They exist only in Metaculus (and on other experimental platforms, like foretold.io)

We first ran 20,000 simulations of a tournament with 100 participants and 30 questions. Results for all 30 questions were drawn from a logistic distribution of mean 50 and standard deviation 10. 

The participants were divided into:

\begin{enumerate}
    \item One perfect predictor, which predicts a logistic with mean 50 and sd 10, i.e., the true underlying distribution. Represented in green.
    \item 10 highly skilled predictors which are somewhat wrong, at various levels of systematic over or underconfidence. They predict a single logistic on each question with mean chosen randomly from 45-55 and standard deviation ranging from \{4,6 ...20,22\}. Represented in orange.
    \item 10 highly skilled predictors, trying to stand out by extremizing some of their forecast distributions. They predict a single logistic on each question with mean chosen randomly from 45-55 and standard deviation 10 for 25 of the questions and standard deviation 5 for a (randomly selected) other 5. Represented in light greenish brown. 
    \item 20 unskilled predictors. They predict a single logistic with means for each question chosen randomly from 35-65 and standard deviation 10. Represented in light blue. 
    \item 20 unskilled, overconfident predictors. They predict a single logistic with means for each question chosen randomly from 35-65 and SD 5. Represented in dark blue.
    \item 39 unskilled, underconfident predictors. They predict a single logistic with means for each question chosen randomly from 35-65 and SD 20. Represented in pink. 
\end{enumerate}

Scoring was according to the Metaculus score formula, though we used the mean of the scores instead of the score of the mean to simplify the process. 

On such a tournament, it would be plausible for only the top 5 forecasters to be rewarded, so we present the probabilities of being in the top 5 for each group.

We find that those who change the position of the mean slightly from the true distribution (i.e., groups 3, 4 and, to a lesser extent, 2), gain a 5\% absolute advantage in terms of getting into the top 5, and a 50\% relative advantage. Each forecaster who does this is $\approx 15\%$ likely to be successful in reaching the top 5. In contrast, the perfect predictor only has a $\approx 10\%$ chance of doing so. 

This can be explained by understanding that a forecaster who wants to reach a discrete ceiling faces a trade-off between the mean of the expected score, and the variance of that score. If the goal is to reach a threshold, increasing the variance at the expense of the mean turns out to in this case be worth it. Graphs which showcase this effect follow. 
\\
\begin{figure}[h!]
    \includegraphics[width=\textwidth]{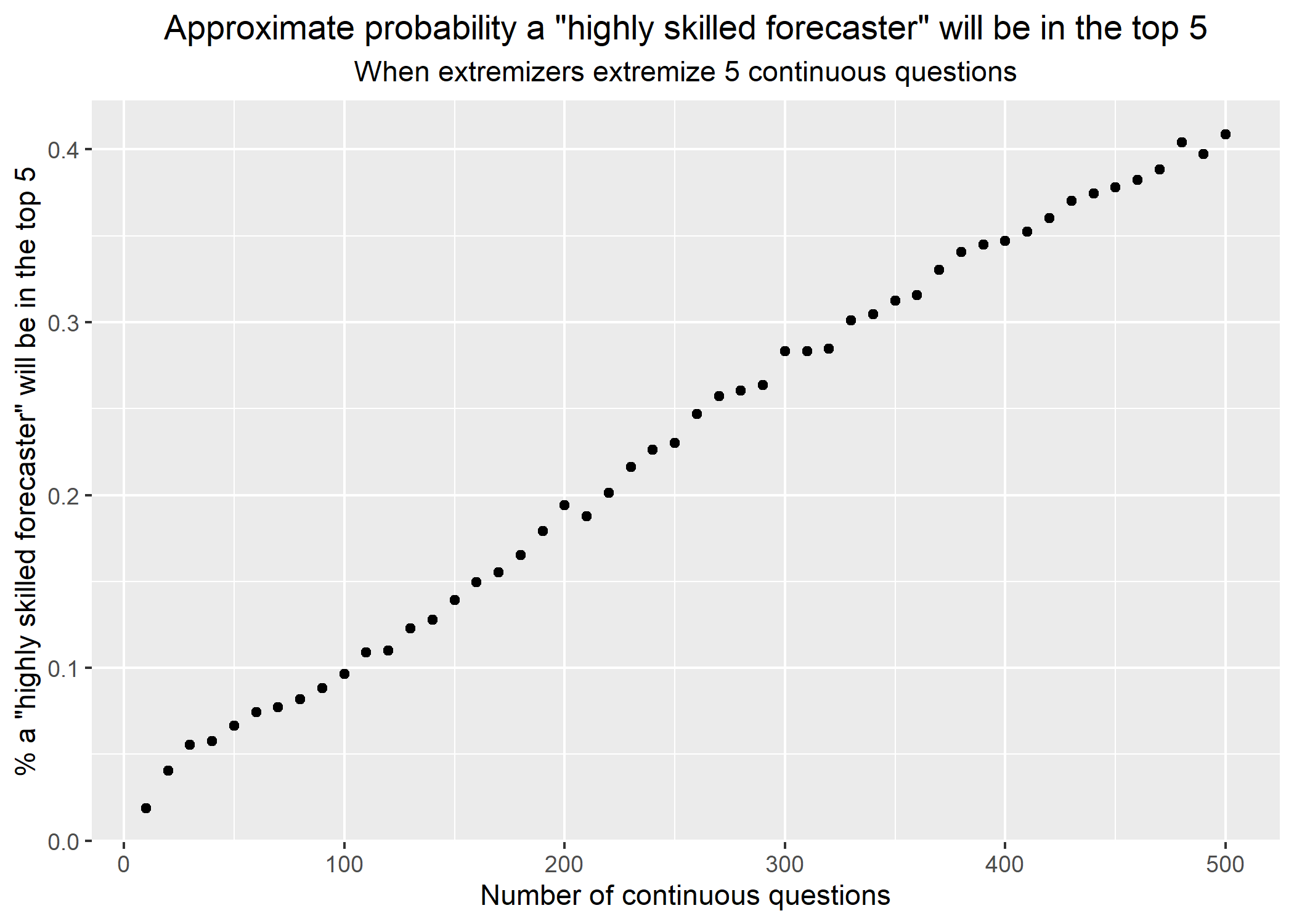}
    \centering
    \caption{Approximate probability (understood as ``frequency during simulations'') that a ``highly skilled forecaster'' will obtain the lowest Brier score, for a tournament with $n$ continuous questions}
\end{figure}





%
%
%
%
%
%

\section{Extremization calculations}
\label{Extremization calculations}

The expected ``participation-rate weighted Brier score'' for a question which closes on the nth day if the event happens, and or in the 5th day if the event hasn't happened by then, will be:

\begin{equation}
\begin{split}
    E[PWBS] &= 0.25 \cdot PWBS(\textnormal{The event happens in the first day}) \\
    &+ 0.75\cdot 0.25 \cdot PWBS(\textnormal{The event happens in the second day}) \\
    &+ 0.75^2\cdot 0.25 \cdot PWBS(\textnormal{The event happens in the third day}) \\
    &+ 0.75^3\cdot 0.25 \cdot PWBS(\textnormal{The event happens in the fourth day}) \\
    &+ 0.75^4\cdot PWBS(\textnormal{The event doesn't happen}) \\
\end{split}
\end{equation}

Now, the integral in (\ref{PWBS definition}) transforms into a simple sum, because probabilities stay constant throughout the day, so

\begin{equation}
\begin{split}
        &PWBS(\textnormal{The event happens in the nth day}) \\
        &= \frac{Brier(p(S_1)) + ... + Brier(p(S_n))}{n}
\end{split}
\end{equation}

Likewise for the event not happening, except that $Brier(p)$ will be equal to $(0-p)^2$ instead of $(1-p)^2$:

\begin{equation}
\begin{split}
        &PWBS(\textnormal{The event doesn't happen}) \\
        &= \frac{Brier(p(S_1)) + ... + Brier(p(S_4))}{4}
\end{split}
\end{equation}
And so, the 
\begin{equation}
\begin{split}
        &E[PWBS] = 0.25 \cdot (1-p(S_1))^2\\
        &+ 0.75\cdot 0.25 \cdot \frac{ (1-p(S_1))^2 + (1-p(S_2))^2}{2} \\
        &+ 0.75^2\cdot 0.25 \cdot \frac{ (1-p(S_1))^2 + (1-p(S_2))^2 + (1-p(S_3))^2 }{3} \\
        &+ 0.75^3\cdot 0.25 \cdot \frac{ (1-p(S_1))^2 + (1-p(S_2))^2 + (1-p(S_3))^2 +(1-p(S_4))^2}{4} \\
        &+ 0.75^4 \cdot \frac{ (0-p(S_1))^2 + (0-p(S_2))^2 + (0-p(S_3))^2 +(0-p(S_4))^2}{4} \\
\end{split}
\end{equation}

Given this, we can calculate the expected ``participation weighted Brier score" for the true probabilities, and the extremized probabilities.

\begin{equation}
    \begin{split}
        E[PWBS(\textnormal{honest probabilities)}] = 0.193 \\
        E[PWBS(\textnormal{extremized probabilities)}] = 0.182
    \end{split}
\end{equation}

\end{appendices}
\end{document}